\documentclass[a4paper,12pt,reqno]{amsart}

\usepackage{hyperref}
\usepackage{amsmath,amssymb,mathrsfs}
\usepackage{graphicx}
\usepackage{multicol}
\usepackage{enumerate}

\newtheorem{theorem}{Theorem}
\newtheorem{lemma}{Lemma}[section]
\newtheorem{proposition}[lemma]{Proposition}

\newtheorem{remark}[lemma]{Remark}

\numberwithin{equation}{section}

\newcommand{\CO}{\mathbb C}
\newcommand{\RE}{\mathbb R}
\newcommand{\NA}{\mathbb N}
\newcommand{\ve}{\varepsilon}
\newcommand{\de}{\delta}
\newcommand{\ga}{\gamma}
\newcommand{\la}{\lambda}

\newcommand{\FF}{\mathfrak{F}}
\newcommand{\GG}{\mathbf{G}}
\newcommand{\TT}{\mathbf{T}}
\newcommand{\VV}{\mathbf{V}}
\newcommand{\RR}{\mathbf{R}}
\newcommand{\HH}{\mathbf{H}}

\renewcommand{\tt}{\mathbf{t}}
\newcommand{\rr}{\mathbf{r}}
\newcommand{\vv}{\mathbf{v}}
\newcommand{\hh}{\mathbf{h}}
\newcommand{\ID}{\mathbf{1}}

\newcommand{\BB}{\mathcal{B}}

\renewcommand{\t}{t}
\newcommand{\V}{V}

\renewcommand{\gg}[1]{\tilde g^{(#1)}}
\newcommand{\trho}[1]{\tilde \rho^{(#1)}}
\newcommand{\ff}[1]{\tilde f^{[#1]}}

\newcommand{\trhont}[1]{\rho^{(#1)}}

\newcommand{\Bnorm}{\mathscr{B}}

\newcommand{\bb}{b}

\renewcommand{\Im}{\operatorname{Im}\,}

\newcommand{\lf}{\left}
\newcommand{\ri}{\right}
\newcommand{\f}{\frac}

\renewcommand{\leqslant}{\leq}

\title[The 3-body problem in dimension one: From short-range to contact interactions]{The three-body problem in dimension one: From short-range to contact interactions}

\author[G. Basti]{Giulia Basti}
\address[G. Basti]{Institut f\"{u}r Mathematik, Universit\"{a}t Z\"{u}rich, Winterthurerstrasse 190, CH-8057 Z\"{u}rich, Switzerland}
\email{giulia.basti@math.uzh.ch} 

\author[C. Cacciapuoti]{Claudio Cacciapuoti}
\address[C. Cacciapuoti]{DiSAT, Sezione di Matematica, Universit\`a dell'Insubria, via Valleggio 11, 22100 Como, Italy}
\email{claudio.cacciapuoti@uninsubria.it	}

\author[D. Finco]{Domenico Finco}
\address[D. Finco]{Facolt\`a di Ingegneria, Universit\`a Telematica
Internazionale Uninettuno,  Corso Vittorio Emanuele II 39, 00186 Roma, Italy}
\email{d.finco@uninettunouniversity.net}

\author[A. Teta]{Alessandro Teta}
\address[A. Teta]{Dipartimento di Matematica G. Castelnuovo, Sapienza Universit\`a di Roma,  Piazzale  Aldo Moro, 5, 00185 Roma, Italy}
\email{teta@mat.uniroma1.it}

\date{}

\thanks{
The authors acknowledge the support of the GNFM Gruppo Nazionale per la Fisica Matematica - INdAM. G.B., C.C., and D.F. acknowledge  the support of the project ``Progetto Giovani GNFM 2016''. 
}

\begin{document}
\begin{abstract}
We consider a Hamiltonian describing three quantum particles in dimension one interacting through  two-body short-range potentials. We prove that, as a suitable scale parameter in the potential terms  goes to zero,  such Hamiltonian converges to one with  zero-range (also called delta or point) interactions. The convergence is understood in norm resolvent sense. The two-body rescaled potentials are of the form $v^{\varepsilon}_{\sigma}(x_{\sigma})= \varepsilon^{-1} v_{\sigma}(\varepsilon^{-1}x_\sigma )$, where $\sigma = 23, 12, 31$ is an index that runs over all the possible pairings of the three particles,  $x_{\sigma}$ is the relative coordinate between two particles, and $\varepsilon$ is the scale parameter. The limiting Hamiltonian is the one formally obtained by replacing the potentials $v_\sigma$ with $\alpha_\sigma \delta_\sigma$, where $\delta_\sigma$ is the Dirac delta-distribution centered on the coincidence hyperplane $x_\sigma=0$ and $\alpha_\sigma = \int_{\mathbb{R}} v_\sigma dx_\sigma$.  To prove the convergence of the resolvents we make use of Faddeev's equations.
\end{abstract}

\maketitle

\begin{footnotesize}
 
 \end{footnotesize}

\vspace{1cm}

\section{Introduction}

In a dilute quantum gas at low temperature the typical wavelength of the particles is usually much larger than the effective range of the two-body interaction. In this regime the system exhibits a universal behavior,  which means that the relevant observables do not depend on the details of the interaction but only on few low-energy parameters, like the scattering length. For the mathematical modeling of these systems it is often convenient to introduce  Hamiltonians where the two-body interaction is replaced by an idealized zero-range or $\delta$ interaction, i.e., an interaction that is nontrivial only when the coordinates $x_i$ and $x_j$ of two particles  coincide. A Hamiltonian of this type is usually constructed as a self-adjoint operator in the appropriate Hilbert space using the theory of self-adjoint extensions. Roughly speaking, one obtains an operator acting as the free Hamiltonian except at the coincidence hyperplanes $\{x_i=x_j\}$, $i<j$, where a suitable boundary condition is satisfied. Many interesting mathematical results in this direction are available, see, e.g., \cite{AGHH05} which addresses mostly the  two-body problem, and \cite{michelangeli-ottolini-rmp17} for a review on the $N$-body problem, mainly in dimension three, and references therein. Here we only remark that these results strongly depend on the dimension $d$ of the configuration space. In particular, for $d=1$ the resulting Hamiltonian is a small perturbation in the sense of the quadratic forms of the free Hamiltonian, for $d=2,3$  the situation is different  and the  Hamiltonian is characterized by  singular boundary conditions  at the coincidence hyperplanes and, finally, for $d>3$ a no-go theorem prevents the construction of a nontrivial zero-range interaction.

The construction of Hamiltonians with zero-range interactions based on the theory of self-adjoint extensions could appear rather abstract from the physical point of view. A more transparent and natural justification  is obtained if one shows that these Hamiltonians are the limit of Hamiltonians with smooth, suitably rescaled two-body potentials. 
In the two-body case, reduced to a one-body problem in the relative coordinate,  such a procedure is well established in all dimensions $d=1,2,3$, see \cite{AGHH05}, while in the case of three or more particles only few results are available (\cite{dellantonio17}). 

In this paper we approach the problem in the simpler case of three particles in dimension one. More precisely, we consider the three-body Hamiltonian 
\[
\HH^{\ve,3} : = - \frac{1}{2m_1} \Delta_1 - \frac{1}{2m_2} \Delta_2 - \frac{1}{2m_3} \Delta_3+ \VV_{12}^\ve +\VV_{23}^\ve +\VV_{31}^\ve = \HH_0^{3} + \sum_{\sigma} \VV_{\sigma}^\ve\,, 
\]
where $m_j$ is the mass of the $j$-th particle and  $\Delta_j$ denotes the one-dimensional Laplacian with respect to the coordinate $x_j$ of the $j$-th particle. We use greek letters $\sigma, \gamma, \ldots$ to denote  an index that runs over the pairs $12$, $23$, and $31$ and, for simplicity, we set $\hbar=1$. Moreover,    $\VV_{\sigma}^\ve$, for $\ve>0$, describes the two-body, rescaled interaction between the particles in the pair $\sigma$, i.e.,  $\VV_{23}^\ve$ denotes the multiplication operator by the rescaled potential $v^{\ve}_{23}(x_2-x_3)= \ve^{-1} v_{23}(\ve^{-1}(x_2-x_3))$ (and similarly for the other two pairs). 

One reasonably expects that for $\ve \to 0$  the above Hamiltonian reduces to  the Hamiltonian formally written as 
\[\HH^{3} : = - \frac{1}{2m_1} \Delta_1 - \frac{1}{2m_2} \Delta_2 - \frac{1}{2m_3} \Delta_3+\alpha_{12} \delta_{12} +\alpha_{23}\delta_{23} +\alpha_{31}\delta_{31} = \HH_0^{3} + \sum_{\sigma} \alpha_\sigma\delta_{\sigma}\,, \]
where $\delta_{23}$ denotes  the Dirac-delta distribution supported on the coincidence plane $\{ x_2=x_3 \}$  of the second and third particle  (and similarly for $\delta_{12}$ and $\delta_{31}$). Here $\delta_{\sigma}$ are   understood as distributions on $\mathcal S(\RE^3)$, and $\alpha_\sigma$ are some fixed real parameters, depending on $v_{\sigma}$, which measure the strength of the interaction. 

In order to study the limiting procedure $\ve \to 0$, it is convenient to work in the center of mass reference frame, so that the Hilbert space of the states of the system reduces to  $L^2(\RE^2)$. We denote by $(x_{\gamma},y_{\ell})$ a generic set  of Jacobi coordinates, 
where $\gamma$ is an index that can assume value over any of the pairs  $12$, $23$, and $31$ and   $\ell$ (more precisely, one should write  $\ell_{\gamma}$) is the companion index of $\gamma$, which means that if $\gamma = 23$ then $\ell =1$ and so on. For example, we have 
\begin{equation*}
 x_{23} = x_2-x_3\;;\quad\quad y_1 =  \frac{m_2 x_2 + m_3x_3}{m_2 + m_3}- x_1\,.  
\end{equation*}

In the center of mass reference frame and using the Jacobi coordinates the approximating Hamiltonian has the form
\begin{equation}\label{Hred}
\HH^{\ve} := - \frac{1}{2m_{\gamma}} \Delta_{x_{\gamma}} - \frac{1}{2\mu_\ell} \Delta_{y_\ell} +\sum_{\sigma} \VV^{\ve}_{\sigma} = \HH_0 +\sum_{\sigma} \VV^{\ve}_{\sigma}\,,  
\end{equation}
where $m_{\gamma}$ is the reduced mass between the particles of the pair $\gamma$, and  $\mu_{\ell}$ is the reduced mass between the particle $\ell$ and the subsystem composed by the two particles of the pair $\gamma$, i.e., 
\begin{equation}\label{masses}
m_{23} = \frac{m_2m_3}{m_2+m_3}\;;\quad\quad \mu_1 = \frac{m_1(m_2+m_3)}{M} \qquad \textrm{with} \quad M = m_1+m_2+m_3,
\end{equation}
and similarly for the other pairs. We shall  assume conditions on the potentials $v_{\sigma}$ such that $\HH^{\ve}$ is a self-adjoint and lower bounded operator in $L^2(\RE^2)$, with a lower bound independent of $\ve$ (see Section \ref{s:fadeq}). The limiting Hamiltonian has the formal expression
\begin{equation}\label{rat}
\HH := \HH_0 +\sum_{\sigma} \alpha_{\sigma} \delta_{\sigma}\,.  
\end{equation}
Its rigorous definition as a self-adjoint, lower bounded operator in $L^2(\RE^2)$ will be given in  Section \ref{s:Halde}. 

Our main result is stated in the following:
\begin{theorem}\label{t:main}
Assume that $v_\sigma \in L^1(\RE,(1+|x|)^s dx)$ for some $s>0 $ and for all $\sigma = 23,31,12$. Moreover set $\alpha_\sigma = \int_{\RE} v_\sigma (x)  \, dx$. Then $\HH^{\ve}$ converges to $\HH$ in norm resolvent sense for $\ve \to 0$. 
\end{theorem}

\begin{remark}
From the proof of the theorem it is clear that larger $s$ gives faster convergence speed, up to $s=1$. More precisely  for all $z\in\CO$ with $\Im z\neq 0 $ one has that
\begin{equation}\label{cori}
\|(\HH^\ve  -z)^{-1}- (\HH -z)^{-1}\|_{\Bnorm(L^2(\RE^2))} \leq C \ve^{\delta} \qquad \forall\, \delta < \min\{1,s\} . 
\end{equation}
\end{remark}


The paper is organized as follows.\\
In Section \ref{s:fadeq} we show that $\HH^{\ve}$ is self-adjoint and lower bounded in $L^2(\RE^2)$, with a lower bound independent of $\ve$. Moreover, we write the resolvent of $\HH^{\ve}$ in the form of   Faddeev's equations in momentum space.\\
In Section \ref{s:Halde} we construct the limiting Hamiltonian $\HH$ as a self-adjoint and lower bounded operator in $L^2(\RE^2)$ and we find a suitable representation (in a form that resembles Faddeev's equations) for the resolvent in momentum space.\\
Section 4 is devoted to the proof of Theorem 1.  In particular, we first prove estimate \eqref{cori} for $z =-\lambda$, with $\lambda >0$ large enough, and then we extend the result to $z \in \CO \setminus \RE$. \\
We conclude the paper with two appendices. In Appendix \ref{app:A} we recall the derivation of Faddeev's equations, and the definitions and basic properties of the operators introduced in Section \ref{s:fadeq}. In Appendix \ref{app:B} we collect several explicit formulae and useful identities, mostly concerning the  operators introduced in Section \ref{s:Halde}. 

In what follows $C$ denotes a generic positive constant, independent of the parameters $\varepsilon$ and $\lambda$. 
\section{The approximating problem \label{s:fadeq}}

We denote by $\BB_0$ the sesquilinear form 
\begin{equation}\label{B0}
 \BB_0(\varphi,\psi)=  \frac{1}{2 m_\ga} \int_{\RE^2} \overline{\partial_{x_\ga} \varphi} \, \partial_{x_\ga} \psi \,  dx_\ga \, dy_{\ell} +
 \frac{1}{2 \mu_\ell} \int_{\RE^2} \overline{\partial_{y_\ell} \varphi} \, \partial_{y_\ell} \psi \,  dx_\ga \, dy_{\ell}
\qquad 
D(\BB_0) :=H^1(\RE^2)\times H^1(\RE^2).
\end{equation}
In what follows, with a slight abuse of notation,  we shall denote by the same symbol the corresponding quadratic form $\BB_0(\psi) \equiv \BB_0(\psi, \psi)$ with domain $H^1(\RE^2)$.

The quadratic form associated to $\HH^\ve$ is  
\[
\BB^\ve(\psi)= \BB_0(\psi) + \sum_{\sigma} \left(\psi,  \VV^\ve_\sigma \psi \right)_{L^2(\RE^2)}  \qquad 
D(\BB^\ve) :=H^1(\RE^2) .
\]

We note the inequality 
\begin{equation}\label{colors2}
\sup_{x\in\RE}\int_\RE dy |\psi(x,y)|^2 \leq  \eta  \left\|\partial_{x}\psi\right\|_{L^2(\RE^2)}^2+ \frac{1}{\eta}\left\|\psi\right\|_{L^2(\RE^2)}^2,
\end{equation}
which holds true for all $\eta>0$ (for a proof, see Eq. \eqref{bolletta} below). 
By Eq. \eqref{colors2}, and by the change of variables $x_\sigma /\ve \to x_\sigma $, it immediately follows that 
\[ \left|\left(\psi,  \VV^\ve_\sigma \psi \right)_{L^2(\RE^2)}\right| 
= \int_\RE dx_\sigma |v_\sigma(x_\sigma)| \int_\RE dy_\ell  |\psi(\ve x_\sigma,y_\ell)|^2 
\leq \|v_\sigma\|_{L^1(\RE)} \left(  \eta  \left\|\partial_{x_\sigma}\psi\right\|_{L^2(\RE^2)}^2+ \frac{1}{\eta}\left\|\psi\right\|_{L^2(\RE^2)}^2\right) 
  \]
for all  $\eta >0 $. Hence
 \[ \left|\sum_{\sigma} \left(\psi,  \VV^\ve_\sigma \psi \right)_{L^2(\RE^2)}  \right| \leq a\, \BB_0(\psi) + b \left\|\psi\right\|_{L^2(\RE^2)}^2
  \]
for some $0<a<1$ and $b>0$,  and by KLMN theorem the form $\BB^\ve$ is closed, semi-bounded, and defines a self-adjoint operator, coinciding with $\HH^\ve$, see, e.g., \cite{exner-blank-havlicek08}. Additionally, $\HH^\ve$ is bounded from below uniformly in $\ve$, i.e., there exists $\la_0>0$ such that $\inf\sigma(\HH^\ve) > -\la_0$ for all $\ve>0$.  
Concerning the proof of Eq. \eqref{colors2}, we note that it follows from   the identity
\begin{equation}\label{bolletta}
\int_{-\infty}^\infty dy_{\ell} \, \left|\psi(x,y_\ell)\right|^2 =  \int_{-\infty}^\infty dy_{\ell} \int_{-\infty}^x dx_{\gamma} \,\partial_{x_\gamma} \left|\psi(x_\gamma,y_\ell)\right|^2,
\end{equation}
and the chain of  inequalities 
\[\partial_{x_\gamma} \left|\psi(x_\gamma,y_\ell)\right|^2 \leq 2 |\partial_{x_\gamma}\psi(x_\gamma,y_\ell)| \,|\psi(x_\gamma,y_\ell)| \leq \eta |\partial_{x_\gamma}\psi(x_\gamma,y_\ell)|^2 + |\psi(x_\gamma,y_\ell)|^2 / \eta.\]

Since it is more convenient to formulate both the approximating and the limiting   problem in Fourier space, in what follows we introduce some notation concerning the variables in momentum space. We remark  that we define the Fourier transform  so as to be unitary in $L^2(\RE^d)$, see Appendix \ref{app:B} for the explicit definition. 

We denote by $k_\sigma$ the conjugate coordinate of $x_\sigma$ and by $p_\ell$ the conjugate coordinate of $y_\ell$. Let $\hat \HH_0$ be the operator unitarily equivalent to $\HH_0$ via Fourier transform and let $\hat \RR_0 (\la) = ( \hat \HH_0 +\la)^{-1}$, $\la > 0$. Both of them act as multiplication operators,  more precisely:
\[
\hat \HH_0 f ( k_\sigma , p_\ell) = \lf( \frac{k^2_\sigma}{2 m_\sigma} + \frac{p_\ell^2}{2\mu_\ell}\ri)
 f ( k_\sigma , p_\ell)\,;\qquad
\hat \RR_0 (\la)  f ( k_\sigma , p_\ell) = 
\lf( \frac{k^2_\sigma}{2 m_\sigma} + \frac{p_\ell^2}{2\mu_\ell} + \la\ri)^{-1} 
 f ( k_\sigma , p_\ell).
\]

For the reader's sake  we recall  that  different pairs of Jacobi coordinates are related by  the following formulae  
\begin{align}
&k_{31} = \frac{m_3M}{(m_2+m_3)(m_3+m_1)} p_1 - \frac{m_1}{m_3+m_1}k_{23}    &;& \qquad 
p_2  = -\frac{m_2}{m_2+m_3}p_1 - k_{23} \label{change1} \\ 
& k_{12} = -\frac{m_2M}{(m_2+m_3)(m_1+m_2)} p_1 - \frac{m_1}{m_1+m_2}k_{23}   &;& \qquad
p_3  = -\frac{m_3}{m_2+m_3}p_1 + k_{23} ,\label{change2} 
\end{align}
where $M$ is the total mass of the system, see Eq. \eqref{masses}. The other changes of coordinates are obtained by permutation of the indices in the formulae above. For example, if for sake of concreteness we fix $\sigma =\{ 23\}$ and $\ell=1$  we have
\begin{equation*}\label{H0four}
\hat \HH_0  f (k_{23},p_1) = \left(\frac{k_{23}^2}{2m_{23}} + \frac{p_1^2}{2\mu_1}\right) f(k_{23},p_1).
\end{equation*}

We can also write functions in the $p$'s coordinates only, for this reason we recall the change of variables 
\begin{equation}\label{change3}
k_{23}= -p_2 -\frac{m_2}{m_2+m_3} p_1\qquad ;\qquad p_1 = p_1.
\end{equation}
In the coordinates $(p_2,p_1)$ we have
\begin{equation} \label{change4}
\hat \HH_0 f (p_2,p_1) =  \left(\frac{p_2^2}{2m_{23}} + \frac{p_2 \cdot p_1}{m_3} + \frac{p_1^2}{2m_{13}}\right) f(p_2,p_1).
\end{equation}
We remark that in the latter formula we abused notation and used the symbol $f$ to denote the same function written in two different systems of coordinates, the $(k_\gamma, p_\ell)$-coordinates and the $p$-coordinates. \\
Analogous changes of coordinates  are obtained by permutations of the indices and by taking into account the identity $p_1+p_2+p_3=0$, for more explicit formulae we refer to \cite{faddeev_book}.
Similar formulae hold for $\hat \RR_0 (\la)$. \\

We introduce some notation before representing $\RR^\ve (\la) = (\HH^\ve + \la)^{-1}$ through Faddeev's equations. Here we always assume $\lambda > 0$ such that $\inf\sigma(\HH^\ve) > -\la$ for all $\ve>0$.  

Denote by $t_\gamma^\ve (\lambda;k_\gamma,k_\gamma')$ the integral kernel in Fourier  transform of the operator $\tt_\gamma^\ve(\lambda): L^2(\RE) \to L^2 (\RE)$ defined in Eq. \eqref{tt2}. One has that 
\begin{equation}\label{tt}
\hat\tt_\gamma^\ve(\lambda) f (k_\gamma) = \int_{\RE}  t_\gamma^\ve (\lambda;k_\gamma,k_\gamma') f(k_\gamma') dk_\gamma'.
\end{equation}
By taking the Fourier transform of Eq. \eqref{tt2} one infers that the kernel $t_\ga^\ve (\la)$ satisfies the following integral equation: 
\begin{equation}\label{t}
t_\gamma^\ve (\lambda;k_\gamma,k_\gamma')  = \frac{1}{\sqrt{2\pi}} \hat v_\gamma(\ve (k_\gamma - k_\gamma'))  
-  \frac{1}{\sqrt{2\pi}} \int_{\RE} dq \, \hat{v}_\gamma(\ve(k_\gamma - q)) \frac{1}{q^2/(2m_\gamma)+\lambda} t_\gamma^\ve (\lambda;q,k_\gamma').
\end{equation}
Hence, by Eq. \eqref{TT2}, in Fourier transform, the operator $\TT_\gamma^\ve(\lambda):  L^2(\RE^2) \to L^2 (\RE^2)$ is given by:
\begin{equation}\label{electric}
\hat\TT_\gamma^\ve(\lambda) f(k_\gamma,p_\ell) = \int_{\RE} dk_\gamma' \, t_\gamma^\ve(\lambda + p_\ell^2/(2\mu_\ell);k_\gamma,k_\gamma')  f(k_\gamma',p_\ell).
\end{equation}
We remark that in what follows, in particular in Eq. \eqref{rho1eq}, we shall rewrite the latter formula for  $\hat\TT_\gamma^\ve(\lambda)$ in the $p$'s coordinates.  Eq. \eqref{rho1eq} below is obtained by   taking into account the changes of variables  \eqref{change1}, \eqref{change2} and \eqref{change3}.

We are now ready to write down Faddeev's equations in explicit form (see Appendix \ref{app:A}, Eqs. \eqref{resolvent}, \eqref{faddeev}, and \eqref{hand}): let $\hat \RR^\ve(\la) $ be the conjugate operator to $ \RR^\ve(\la) $ then we have
\begin{equation}\label{horse}
\hat \RR^\ve(\lambda) f = \hat \RR_0(\lambda) f + \hat  \RR_0(\la)\sum_{m=1}^3\rho^{(m),\ve}(\la).
\end{equation}
where the functions  $\rho^{(m),\ve}(\la)$ satisfy the system of equations obtained by  permuting indices in  
\begin{equation} \label{shepard}
\rho^{(1),\ve}(\la) =    -  \hat\TT_{23}^\ve (\la)\hat \RR_0(\lambda) f  \\ 
- \hat \TT_{23}^\ve (\la) \hat \RR_0(\lambda)\rho^{(2),\ve}(\lambda)
-  \hat\TT_{23}^\ve (\la) \hat \RR_0(\lambda)\rho^{(3),\ve}(\lambda).
\end{equation}
In the coordinates $(p_2,p_1)$ Eq.  \eqref{shepard} reads:
\begin{equation}\label{rho1eq}\begin{aligned}
\rho^{(1),\ve}(\la;q,p) = 
& - \int_{\RE} dq' \, \frac{t_{23}^\ve(\lambda+p^2/(2\mu_1);-q-\frac{m_2}{m_2+m_3}p,-q'-\frac{m_2}{m_2+m_3}p)}{\frac{{q'}^2}{2m_{23}} + \frac{q' \cdot p}{m_3} + \frac{p^2}{2m_{13}}+\la} f(q',p)
\\
& -  \int_{\RE} dq' \, \frac{t_{23}^\ve (\lambda+p^2/(2\mu_1);-q-\frac{m_2}{m_2+m_3}p,-q'-\frac{m_2}{m_2+m_3}p)}{\frac{{q'}^2}{2m_{23}} + \frac{q' \cdot p}{m_3} + \frac{p^2}{2m_{13}}+\la} \rho^{(2),\ve}(\la;-p-q',q')  \\ 
& - \int_{\RE} dq' \, \frac{t_{23}^\ve (\lambda+p^2/(2\mu_1);-q-\frac{m_2}{m_2+m_3}p,q'+\frac{m_3}{m_2+m_3}p)}{\frac{{q'}^2}{2m_{23}} + \frac{  q' \cdot p}{m_2} + \frac{p^2}{2m_{12}}+\la} \rho^{(3),\ve}(\la;p,q')  .
\end{aligned}\end{equation}
We remark that the functions $\rho^{(m),\ve}(\la)$ are always understood to be written in their ``natural'' variables, i.e., $\rho^{(1),\ve}(\la)=\rho^{(1),\ve}(\la;p_2,p_1)$, $\rho^{(2),\ve}(\la)=\rho^{(2),\ve}(\la;p_3,p_2)$, and $\rho^{(3),\ve}(\la)=\rho^{(3),\ve}(\la;p_1,p_3)$.

\section{The limiting problem \label{s:Halde}}
In this section we discuss the rigorous definition of the Hamiltonian $\HH$ describing  three particles interacting through contact interactions and formally written  as in Eq. \eqref{rat}.
%

We shall denote by $\pi_{\gamma}$ the coincidence line (hyperplane) of the particles in the pair $\gamma$, i.e., in the Jacobi coordinates $(x_{\gamma},y_l)$, $\pi_{\gamma}$ is identified by $x_{\gamma} = 0.$ The hyperplanes $\pi_\gamma$ identify six regions $\Gamma_r,$ $r=1,\dots 6.$ For the sake of clarity we write explicitly the definition of $\Gamma_r$ in the coordinates $(x_{23},y_1),$
 obviously we could have equivalently used any other pair of Jacobi coordinates.
\[
	\begin{aligned}
		\Gamma_1 &=\left\{(x_{23},y_1)\bigg|\;
		\begin{gathered} x_{23}\geq0,\\
			y_1<-\frac{m_3}{m_2+m_3}x_{23}
		\end{gathered}\right\},
		& \Gamma_2 &=\left\{(x_{23},y_1)\bigg|\;
		\begin{gathered}
			x_{23}\geq0,\\
			-\frac{m_3}{m_2+m_3}x_{23}<y_1<\frac{m_2}{m_2+m_3}x_{23}
		\end{gathered}\right\},\\
		\Gamma_3 &=\left\{(x_{23},y_1)\bigg|\;
		\begin{gathered}
			x_{23}\geq0,\\
			y_1>\frac{m_2}{m_2+m_3}x_{23}
		\end{gathered}\right\}	,& \Gamma_4 &=\{(x_{23},y_1)|\, (-x_{23},-y_1)\in \Gamma_1\},\\
		\Gamma_5&=\{(x_{23},y_1)|\, (-x_{23},-y_1)\in \Gamma_2\},& \Gamma_6 &=\{(x_{23},y_1)|\, (-x_{23},-y_1)\in \Gamma_3\}.
	\end{aligned}
\] 

For any function $\psi\in H^s(\RE^2)$ with $s> 1/2$ we denote by $\psi|_{\pi_\gamma}$ its trace on the hyperplane $\pi_\gamma$ and we recall that the map   $\psi \to \psi|_{\pi_\gamma}$ extends to a continuous  one from $H^{s} (\RE^n)$ to $H^{s-1/2}(\RE^{n-1})$ for any $n\in\NA$ and $s> 1/2$. Sometimes, when we need to make explicit the dependence of $\psi \to \psi|_{\pi_\gamma}$  on the coordinate $y_\ell,$ we shall simply write $\psi \to \psi|_{\pi_\gamma}(y)$ omitting the suffix $\ell$ when no misunderstanding is possible. 


To give a rigorous definition of the operator $\HH$ we start with a natural choice of the quadratic form: since the potential $\alpha_\gamma \delta_\gamma$ is supported by the hyperplane $\pi_\gamma$, we set 
\[
\BB(\varphi,\psi):= \BB_0(\varphi,\psi) + \sum_{\sigma} \alpha_\sigma \left(\varphi|_{\pi_\sigma}, \psi|_{\pi_\sigma}\right)_{L^2(\pi_\sigma)},  \qquad D(\BB) := H^1(\RE^2)\times H^1(\RE^2)
\]
where the definition of $\BB_0$ was  given in Eq. \eqref{B0}.
With a slight abuse of notation, we denote by the same letter the corresponding quadratic form: 
\[
\BB(\psi):= \BB_0(\psi) + \sum_{\sigma} \alpha_\sigma \left\|\psi|_{\pi_\sigma}\right\|_{L^2(\pi_\sigma)}^2  \qquad D(\BB) :=H^1(\RE^2).
\]
\begin{remark}
By Eq. \eqref{colors2}, it immediately follows that 
\[ \sum_{\sigma}| \alpha_\sigma| \left\|\psi|_{\pi_\sigma}\right\|_{L^2(\pi_\sigma)}^2  \leq a\, \BB_0(\psi) + b \left\|\psi\right\|_{L^2(\RE^2)}^2
  \]
for some $0<a<1$ and $b>0$, hence by KLMN theorem the form $\BB$ is closed, semi-bounded, and defines a self-adjoint operator bounded from below, see also \cite{brasche-exner-kuperin-jmaa94}. 
\end{remark}

We denote by $\Pi$ the union of the hyperplanes $\pi_{12}$, $\pi_{23}$, $\pi_{31}$: $\Pi = \cup_{\sigma} \pi_\sigma $.\\
Moreover we denote by  $ \left[\partial_{x_\gamma}\psi\right]_{\pi_\gamma} $ the jump of the normal derivative of the function $\psi$ across the plane $\pi_\gamma$, i.e., 
\[\left[\partial_{x_\gamma}\psi\right]_{\pi_\gamma}  \equiv \left[\partial_{x_\gamma}\psi\right]_{\pi_\gamma} (y_\ell) := 
\lim_{\eta\to0^+}\left(\partial_{x_\gamma}\psi (\eta,y_\ell) - \partial_{x_\gamma}\psi(-\eta,y_\ell)\right).  \]

\begin{theorem}The self-adjoint operator associated to the closed and semi-bounded quadratic form $\BB$ is 
\begin{gather}
D(\HH):=\left\{ \psi \in H^2(\RE^2\backslash \Pi) \cap H^1(\RE^2) | \; \left[\partial_{x_\gamma}\psi\right]_{\pi_\gamma} = 2m_\gamma\alpha_\gamma \psi|_{\pi_\gamma}\quad \forall \; \gamma\right\}\label{HHalde1}\\
\HH \psi = \HH_0\psi \qquad \text{on $\RE^2\backslash \Pi$}.\label{HHalde2}
\end{gather}
\end{theorem}

\begin{proof}
	According to the general theory the operator associated to $\BB$ is defined by
	\begin{gather*}
		D(\HH):=\{\psi\in D(\BB)|\, \exists f\in L^2(\RE^2)\text{ s.t. } \BB(\varphi,\psi)=(\varphi,f)\;\forall \varphi\in D(\BB)\}\\
		\HH\psi=f.
	\end{gather*}
	Let $\varphi\in C^\infty_0(\Gamma_1)\subset D(\BB)$ and $\psi\in D(\HH).$ Then
	\[
		(\varphi,f)=\BB(\varphi,\psi)=\BB_0(\varphi,\psi)
	\]
	hence $\psi\in H^2(\Gamma_1)$ and $f=\HH_0\psi$ in $\Gamma_1.$ Repeating the argument  for $\Gamma_r,\,r=2,\dots, 6$ we conclude $\psi\in D(\HH)$ implies $\psi\in H^2(\RE^2\backslash \Pi)$ and $f=\HH_0\psi$ on $\RE^2\backslash \Pi.$ This proves Eq. \eqref{HHalde2}.
	
	It remains to show the validity of the boundary conditions in Eq. \eqref{HHalde1}: $[\partial_{x_\gamma}\psi]_{\pi_\gamma}=2m_\gamma \alpha_\gamma\psi|_{\pi_\gamma}$ $\forall\gamma$. To this end we consider $\varphi\in C^\infty_0(\Gamma_6\cup \Gamma_1)\in D(\BB)$. Using for definiteness the coordinates $(x_{23},y_1),$ we have
	\begin{equation}\label{BGamma1}
		(\varphi,f)=\BB(\varphi,\psi)=\BB_0(\varphi,\psi)+\alpha_{23}\int_{-\infty}^0 dy_1 \overline{\varphi\big|_{\pi_{23}}(y_1)}\psi\big|_{\pi_{23}}(y_1).
	\end{equation}
	Let  $\pi_{23,\delta}=\{(x_{23},y_1):|x_{23}|<\delta,y_1<0\}$, $\Gamma_6^{\delta}=\Gamma_6\backslash \pi_{23,\delta},\Gamma_1^{\delta}=\Gamma_1\backslash \pi_{23,\delta}.$ Then, by Eq. \eqref{HHalde2} it follows
	\begin{equation}\label{BGamma2}
		(\varphi,f)=\lim_{\delta\to0}\biggl[\int_{\Gamma_6^\delta}dx_{23}dy_1\,\overline{\varphi}\, \HH_0\psi+\int_{\Gamma_1^\delta}dx_{23}dy_1\,\overline{\varphi}\,\HH_0\psi\biggr].
	\end{equation}
	On the other hand, 
\begin{equation*}
			\BB_{0}(\varphi,\psi)=\lim_{\delta\to0}  \int_{\Gamma_6^\delta\cup \Gamma_1^\delta} \frac{1}{2 m_{23}} \overline{\partial_{x_{23}} \varphi} \, \partial_{x_{23}} \psi  +
 \frac{1}{2 \mu_1}  \overline{\partial_{y_1} \varphi} \, \partial_{y_1} \psi \,  dx_{23} \, dy_{1}.
	\end{equation*}
Integrating by parts and taking the limit for $\delta\to0$ on the boundary term, one obtains 
	\begin{equation}\label{BGamma3}
			\BB_{0}(\varphi,\psi)=\lim_{\delta\to0} \biggl[\int_{\Gamma_6^\delta}dx_{23}dy_1\overline{\varphi}\,\HH_0\psi+\int_{\Gamma_1^\delta}dx_{23}dy_1\overline{\varphi}\,\HH_0\psi\biggr] - \frac{1}{2m_{23}}\int_{-\infty}^0 dy_1\,\overline{\varphi\big|_{\pi_{23}}(y_1)}\,[\partial_{x_{23}}\psi]_{\pi_{23}}(y_1).
	\end{equation}
	By Eqs. \eqref{BGamma1},\eqref{BGamma2},\eqref{BGamma3} we conclude 
	\[
		\int_{-\infty}^0 dy_1\overline{\varphi\big|_{\pi_{23}}(y_1)}\biggl[\alpha_{23}\psi\big|_{\pi_{23}}(y_1)-\frac{1}{2m_{23}} [\partial_{x_{23}}\psi]_{\pi_{23}}(y_1)
		\biggr]=0 \qquad \forall \varphi\in C^{\infty}_0(\Gamma_6\cup\Gamma_1).
	\]
	Hence 
	\[
		\left[\partial_{x_\gamma}\psi\right]_{\pi_{23}} = 2m_{23}\alpha_{23}\, \psi|_{\pi_{23}}
	\]
	on $\pi_{23}^-=\{(x_{23},y_1): x_{23}=0,\,y_1<0\}$. Repeating the argument for $\Gamma_i\cup\Gamma_j, \,i<j$ we conclude the proof.
\end{proof}
%

\newcommand{\qq}{\boldsymbol{q}}
\newcommand{\MM}{\mathbf{M}} 
\renewcommand{\AA}{\mathbf{A}}
\renewcommand{\ID}{\mathbf{I}}

In the following we find the expression of the resolvent operator $(\HH+\lambda)^{-1}$ for $\la>0$  such that $\inf\sigma(\HH) > -\la$.  First we introduce several operators. Let 
\[\breve  \GG(\lambda) :    L^2(\RE^2)\to  L^2(\pi_{23})\oplus L^2(\pi_{31})\oplus L^2(\pi_{12}) \]
\[\breve  \GG(\lambda) :=(\breve G_{23}(\lambda),\breve G_{31}(\lambda),\breve G_{12}(\lambda))\]
with $\breve G_{\gamma}(\lambda): L^2(\RE^2) \to L^2 (\pi_\gamma)$ defined by 
\begin{equation}\label{breveG}
\breve G_{\gamma}(\lambda) f := \RR_0(\lambda) f|_{\pi_{\gamma}}.
\end{equation}
 Let 
\[ \GG(\lambda) :    L^2(\pi_{23})\oplus L^2(\pi_{31})\oplus L^2(\pi_{12}) \to  L^2(\RE^2) \]
\[\GG(\lambda) :=\breve \GG(\lambda)^*.\]
Hence,  for $\qq=(q^{(1)},q^{(2)},q^{(3)})\in L^2(\pi_{23})\oplus L^2(\pi_{31})\oplus L^2(\pi_{12})$ one has 
\begin{equation*}
\GG(\lambda)\qq=\sum_{\gamma}G_\gamma(\lambda)q^{(\ell)}
\end{equation*}
where $G_\gamma(\lambda): L^2(\pi_{\gamma})\to L^2(\RE^2)$  is the adjoint of $\breve G_{\gamma}(\lambda)$. We note that the action of $G_{\gamma}(\lambda)$ is formally given by 
\begin{equation*}\label{Gg}
	G_\gamma(\lambda)q^{(\ell)}=\RR_0(\la) (q^{(\ell)}\delta_{\gamma}).
\end{equation*}
We refer to  $\GG(\lambda)\qq$ as the potential produced by the charges $\qq$. 
 Note  also that, as a matter of fact, the spaces $L^2(\pi_\gamma)$ can be identified with $L^2(\RE,dy_\ell)$. Finally, we introduce two matrix operators acting on $L^2(\pi_{23})\oplus L^2(\pi_{31})\oplus L^2(\pi_{12}).$ The operator $\MM$ defined by 
 \[(\MM(\lambda))_{\gamma\sigma}:= M_{\gamma\sigma}(\lambda) ;\qquad M_{\gamma\sigma}(\lambda) : L^2(\pi_\sigma) \to L^2(\pi_\gamma)\]
with 
\begin{equation*}\label{Mgs}
M_{\gamma\sigma}(\lambda) q :=  G_\sigma (\lambda) q|_{\pi_\gamma}, \qquad q\in L^2(\pi_\sigma),
\end{equation*}
and the constant matrix $\AA$ with components
\[
	A_{\gamma\sigma} = 
	\left\{
	\begin{aligned}
		&\alpha_\gamma\qquad & \gamma = \sigma \\ 
		&0& \gamma\neq \sigma.
	\end{aligned}\right.
\]
Denote moreover by  $\ID$ the identity operator in $L^2(\pi_{23})\oplus L^2(\pi_{31})\oplus L^2(\pi_{12}) $. 

\begin{theorem}
For all $\lambda>0$ sufficiently large one has
\[
	\RR(\lambda) =(\HH+\lambda)^{-1}= \RR_0(\lambda) - \GG(\lambda) \left(\ID + \AA \MM(\lambda)\right)^{-1}\AA \breve\GG(\lambda).
\]
\end{theorem}

\begin{proof}First we remark that $\HH$ is a semi-bounded operator hence its resolvent $\RR(\lambda)$ is a bounded operator  for all  $\lambda>0$ such that $\inf\sigma(\HH) > -\la$. Let $f\in L^2(\RE^2).$ We want to show that the unique solution of
\begin{equation}\label{eq:res}
	(\HH+\lambda)\psi=f
\end{equation}
is given by
\begin{equation}\label{eq:psi}
	\psi=\RR_0(\lambda)f+\GG(\lambda)\qq
\end{equation}
where $\qq\in L^2(\pi_{23})\oplus L^2(\pi_{31})\oplus L^2(\pi_{12})$ is
\begin{equation}\label{eq:xi}
	\qq=-(\ID+\AA\MM(\lambda))^{-1}\AA \breve \GG(\lambda)f. 
\end{equation}
First we show that $(\ID+\AA\MM(\lambda))$ is invertible. Let $q\in L^2(\pi_\gamma)$, recalling Eq. \eqref{night1} and by the unitarity of the Fourier transform, we get 
\[
	\left\|M_{\gamma\gamma}(\lambda)q\right\|^2_{L^2(\RE)}=\frac{m_\gamma}{2}\int_{\RE} dp_\ell\left|\frac{\hat{q}(p_\ell)}{\sqrt{\frac{p_\ell}{2\mu_\ell}+\lambda}}\right|^2\leq C\,\frac{\|q\|^2_{L^2(\RE)}}{\lambda}.
\]
On the other hand, if $q\in L^2(\pi_{\gamma'})$, by Eq. \eqref{night2} (see also Eq. \eqref{night3}), by the unitarity of the Fourier transform and by Cauchy-Schwarz inequality, we get 
\[
\left\|M_{\gamma\gamma'}(\lambda)q\right\|^2_{L^2(\RE)} \leq \frac{\|q\|^2_{L^2(\RE)}}{(2\pi)^2} \int_{\RE^2} dp_\ell\, dp_{\ell'} \frac{1}{\left|\frac{p_\ell^2}{2m_\gamma}+\frac{p_{\ell}\cdot p_{\ell'}}{2m_j}+\frac{p_{\ell'}^2}{2m_{\gamma'}}+\lambda\right|^2}\leq C\,\frac{\|q\|^2_{L^2(\RE)}}{\lambda},\quad \gamma\neq\gamma',
\]
$\ell'$ denoting the companion index of $\gamma'$ and $j\neq\ell,\ell'$. The latter inequality can be easily proved by scaling. We conclude that for $\lambda>0$ sufficiently large  one has $\|\AA\MM(\la)\|<1$,  hence $(\ID+\AA\MM(\la))$ is invertible.\\
It remains to show that $\psi$ defined by Eq. \eqref{eq:psi} is the solution of Eq. \eqref{eq:res}. First note that  $\psi\in H^1(\RE^2)\cap H^2(\RE^2\backslash\Pi)$ as a direct consequence of Eq. \eqref{eq:psi} and of Eq. \eqref{pigs}.  Moreover from Eq. \eqref{pigs} it is easy to convince oneself of the fact  that 
\[
	\bigl[\partial_{x_\gamma}G_\gamma(\lambda)q^{(\ell)}\bigr]_{\pi_\gamma}=-2m_\gamma q^{(\ell)}, \qquad\bigl[\partial_{x_\sigma}G_\gamma(\lambda)q^{(\ell)}\bigr]_{\pi_\sigma}=0 \quad\sigma\neq\gamma,
\]
 since  $[\partial_{x_{\gamma}}\RR_0(\lambda)f]_{\pi_{\gamma}}=0$,  one infers
\[
	[\partial_{x_\gamma}\psi]_{\pi_\gamma}=-2m_{\gamma}q^{(\ell)}.
\]
Using now Eqs. \eqref{eq:xi} and \eqref{eq:psi} one has
\begin{equation}\label{eq:xil}
	q^{(\ell)} = - \alpha_\gamma \breve G_\gamma(\lambda) f -\alpha_\gamma \sum_{\gamma' }M_{\gamma\gamma' }(\lambda)q^{(\ell')}=-\alpha_\gamma\psi|_{\pi_\gamma}.   
\end{equation}
Hence, $[\partial_{x_{\gamma}}\psi]_{\pi_\gamma}=2\alpha_\gamma m_\gamma \psi|_{\pi_\gamma}$ and $\psi$ belongs to $D(\HH).$ Moreover, by Eq. \eqref{pigs}, it follows that  $G_\gamma(\lambda)q^{(\ell)}$ satisfies
\[
	(\HH_0+\lambda)G_\gamma(\lambda)q^{(\ell)}= 0 \qquad\text{on $\RE^2\backslash\Pi$}.
\]
Recalling Eq. \eqref{HHalde2} the equation above implies
\[
	(\HH+\lambda)\psi=f 
\]
which concludes the proof.
\end{proof}

In the following we explicitly write the equation for the charges $\qq$ in momentum space. The Fourier transform of  Eq. \eqref{eq:xil}, taking into account the formulae collected in Appendix \ref{app:B}, gives 
\begin{multline*}
		\left(1+\frac{\alpha_{23}\sqrt{2m_{23}}}{2\sqrt{\frac{p_1^2}{2\mu_1}+\lambda}}\right)\hat q^{(1)}(p_1)=-\frac{\alpha_{23}}{\sqrt{2\pi}}\int_{\RE} dp_2\,\frac{1}{\frac{p_2^2}{2m_{23}}+\frac{p_2\cdot p_1}{m_3}+\frac{p_1^2}{2m_{31}}+\lambda}f (p_2,p_1)\\
		-\frac{\alpha_{23}}{2\pi}\int_{\RE} dp_2\,\frac{1}{\frac{p_2^2}{2m_{23}}+\frac{p_2\cdot p_1}{m_3}+\frac{p_1^2}{2m_{31}}+\lambda}\hat q^{(2)}(p_2)	-\frac{\alpha_{23}}{2\pi}\int_{\RE} dp_3\,\frac{1}{\frac{p_3^2}{2m_{23}}+\frac{p_3\cdot p_1}{m_2}+\frac{p_1^2}{2m_{12}}+\lambda}\hat q^{(3)}(p_3).
\end{multline*}
Two similar equations are obtained by permutation of the indices. We note that with a slight abuse of notation we denoted by the same symbol  the function $f$ and its Fourier transform. 

Define $\xi^{(\ell)}(p_\ell)=\hat{q}^{(\ell)}(p_\ell)/\sqrt{2\pi}$ and set 
\begin{equation*}\label{eq:gamma}
 \tau_{\gamma}(\la):= \frac{1}{2\pi}\frac{\alpha_{\gamma}}{1+\alpha_{\gamma} \sqrt{\frac{m_{\gamma}}{2\la}}} \qquad \alpha_{\gamma} \in \RE.
\end{equation*}
Then $\xi^{(\ell)}(p_\ell)$ satisfy the system of equations  (in what follows we make explicit the dependence of $\xi^{(\ell)}$ on $\lambda$)
\begin{equation}\label{eq:stm}
	\begin{aligned}
	\xi^{(1)}(\la;p)=&- \int_{\RE} dq'\,\frac{\tau_{23}\bigl(\lambda+\frac{p^2}{2\mu_1}\bigr)}{\frac{{q'}^2}{2m_{23}}+\frac{q'\cdot p}{m_3}+\frac{p^2}{2m_{31}}+\lambda}f (q',p)-\int_{\RE} dq'\,\frac{\tau_{23}\bigl(\lambda+\frac{p^2}{2\mu_1}\bigr)}{\frac{{q'}^2}{2m_{23}}+\frac{q'\cdot p}{m_3}+\frac{p^2}{2m_{31}}+\lambda}\xi^{(2)}(\la;q')\\
		          &-\int_{\RE} dq'\,\frac{\tau_{23}\bigl(\lambda+\frac{p^2}{2\mu_1}\bigr)}{\frac{{q'}^2}{2m_{23}}+\frac{q'\cdot p}{m_2}+\frac{p^2}{2m_{12}}+\lambda}\xi^{(3)}(\la;q').
	\end{aligned}
\end{equation}
and two more equations obtained by permutation of indices.

We conclude this section with the proof of a bound on the $L^2$-norm of the functions $\xi^{(j)}$, $j = 1,2,3$, see Prop. \ref{p:xi-bound} below. 
\begin{remark} \label{r:tau}
For any $\alpha_{\gamma} \in \RE$, there exists $\tilde{\lambda}>0$ such that, for all $\lambda>\tilde{\lambda}$, one has that
\[|\tau_\gamma (\lambda)| \leq \frac{|\alpha_\gamma|}{\pi}. \]
To see that this is indeed the case: if $\alpha_\gamma  \geq 0$ one has $\tau_\gamma (\lambda)\leq \frac{\alpha}{2\pi}$ for all $\lambda>0$; if $\alpha<0$ take $\tilde{\lambda} = 2\alpha_\gamma^2m_\gamma$. 
\end{remark}
\begin{remark}\label{r:away2}
 We note that 
\[\frac{{q}^2}{2m_{23}} + \frac{q \cdot p}{m_3} + \frac{p^2}{2m_{31}} \geq \frac{q^2+p^2}{2 \max\{m_1,m_2\}}.  \]
So that, by  setting $C_{12} := 2\max\{m_1,m_3\}$, we get 
\[\frac{1}{\frac{{q}^2}{2m_{23}} + \frac{q \cdot p}{m_3} + \frac{p^2}{2m_{31}}+\la} \leq \frac{C_{12}}{q^2+p^2 +C_{12}\la }.  \]
In what follows we shall often use, without further warning,  the latter inequality (or similar ones obtained by permutation of the indices). Moreover, we shall use the identity 
\[\int_\RE \frac{1}{(s^2+\eta)^b}ds = \frac{C_b}{\eta^{b-\frac12}} \qquad \eta>0, \;b >1/2.\]
\end{remark}
\begin{proposition}\label{p:xi-bound}
For all $\la >0$ sufficiently large one has  
\begin{equation}\label{xi-bound}\|\xi^{(j)}(\la)\|_{L^2(\RE)} \leq C\, \|f\|_{L^2(\RE^2)} \qquad j=1,2,3.
\end{equation}
\end{proposition}
\begin{proof}
From Eq. \eqref{eq:stm}, we have that 
\begin{equation*}
	\begin{aligned}
\int_{\RE} dp \, |\xi^{(1)}(\la;p)|^2 \leq  & C \bigg[\int_{\RE} dp  \left(\int_{\RE} dq'\,\frac{1}{{q'}^2+p^2+C_{12}\lambda}|f(q',p) |\right)^2\\ 
	+\int_{\RE} dp  \bigg(\int_{\RE} dq'\, & \frac{1}{{q'}^2+p^2+C_{12}\lambda}|\xi^{(2)}(\la;q')| \bigg)^2+\int_{\RE} dp \left(\int_{\RE} dq'\,\frac{1}{{q'}^2+p^2+C_{31}\lambda}|\xi^{(3)}(\la;q')|\right)^2 \bigg]\\
	\leq & C\bigg[ \int_{\RE} dp \, \frac{1}{(p^2+C_{12}\lambda)^{\frac32}}  \int_{\RE} dq'\,|f (q',p)|^2  \\ 
	+ \int_{\RE} dp \, & \frac{1}{(p^2+C_{12}\lambda)^{\frac32}}  \int_{\RE} dq'\,|\xi^{(2)}(\la;q')|^2+\int_{\RE} dp  \frac{1}{(p^2+C_{31}\lambda)^{\frac32}}  \int_{\RE} dq'\, |\xi^{(3)}(\la;q')|^2\bigg].
	\end{aligned}
\end{equation*}
Here, in the first inequality, we took into account Rem. \ref{r:tau} and Rem. \ref{r:away2}, in the second inequality we used Cauchy-Schwarz inequality and again Rem. \ref{r:away2}. 
Hence
\begin{equation*}
\|\xi^{(1)}(\la)\|^2_{L^2(\RE)} \leq  \frac{C}{\lambda^{\frac32}}\|f\|^2_{L^2(\RE^2)} +  \frac{C}{\lambda} \left(\|\xi^{(2)}(\la)\|^2_{L^2(\RE)} + \|\xi^{(3)}(\la)\|^2_{L^2(\RE)}\right). 
\end{equation*}
Similar inequalities are obtained by permutation of the indices, i.e., 
 \begin{equation*}
\|\xi^{(2)}(\la)\|^2_{L^2(\RE)} \leq  \frac{C}{\lambda^{\frac32}}\|f\|^2_{L^2(\RE^2)} +  \frac{C}{\lambda} \left(\|\xi^{(1)}(\la)\|^2_{L^2(\RE)} + \|\xi^{(3)}(\la)\|^2_{L^2(\RE)}\right),
\end{equation*}
\begin{equation*}
\|\xi^{(3)}(\la)\|^2_{L^2(\RE)} \leq  \frac{C}{\lambda^{\frac32}}\|f\|^2_{L^2(\RE^2)} +  \frac{C}{\lambda} \left(\|\xi^{(1)}(\la)\|^2_{L^2(\RE)} + \|\xi^{(2)}(\la)\|^2_{L^2(\RE)}\right). 
\end{equation*}
Summing up all the inequalities we obtain 
\[
\left(1-\frac{C}{\la}\right) \sum_{j=1}^3 \|\xi^{(j)}(\la)\|^2_{L^2(\RE)} \leq  \frac{C}{\lambda^{\frac32}}\|f\|^2_{L^2(\RE^2)} .
\]
For $\lambda$ large enough the latter inequality implies $ \sum_{j=1}^3 \|\xi^{(j)}(\la)\|^2_{L^2(\RE)} \leq  C\|f\|^2_{L^2(\RE^2)} $, which in turn implies Bound  \eqref{xi-bound}.
\end{proof}

\section{Proof of Main Theorem}

In this section we prove Theorem \ref{t:main}. As a preliminary result we prove an a priori estimate on $\t_\gamma^\ve(\lambda)$ and a bound on $\t_\gamma^\ve(\lambda) - \tau_\gamma (\lambda)$. 
\begin{lemma}\label{l:tvet}
Assume that $ v_\ga \in L^1(\RE,(1+|x|)^b dx)$ for some $0<b < 1$ and for all $\gamma = 23,31,12$. Moreover set $\alpha_\gamma = \int_{\RE}  v_\gamma \, dx$. Then  for all $\la>0$ sufficiently large one has
  \begin{equation}\label{tve-bound}
\sup_{k,k'\in\RE}|\t^\ve_{\gamma}(\lambda;k,k')|\leq   C,
\end{equation}
and 
\begin{equation}\label{tvet}
\sup_{k,k'\in\RE}\frac{|\t_\gamma^\ve(\lambda;k,k') - \tau_\gamma (\lambda)|}{|k|^{\bb}+{|k'|}^{\bb}+1} \leq   C \, \ve^{\bb},
\end{equation}
where $0<\ve<1$. 
\end{lemma}
\begin{remark} Obviously the assumption on the potentials $v_\gamma$ is satisfied whenever $ v_\ga \in L^1(\RE,(1+|x|)^s dx)$ for some $s>0$. We note that  the speed of converges in  Bound \eqref{tvet} improves only up to $s=1$, in particular it does not exceed  $\ve^{\delta}$, $\delta < \min\{1,s\}$.
\end{remark}
\begin{proof}[Proof of Lemma \ref{l:tvet}] Denoting by $\hat v_\gamma(k)$ the Fourier transform of $v_{\gamma}(x)$, we note that, since $v_\ga \in L^1(\RE)$, 
\begin{equation}\label{nothing1}
\sup_{k\in\RE}  | \hat v_\gamma(k)|  \leq C.
\end{equation}
From Eqs. \eqref{t} and \eqref{nothing1} we  infer 
\[
|\t_\gamma^\ve(\lambda;k,k')| \leq C + \frac{C}{\sqrt \lambda }\sup_{p\in\RE} |\t_\gamma^\ve(\lambda;p,k')|.
\]
By taking the supremum over $k$ and $k'$, and up to choosing $\lambda$ large enough, the latter bound implies the a priori estimate \eqref{tve-bound}. \\ 

To prove Bound \eqref{tvet}, we start by noting that, since  $v_\ga\in L^1(\RE,(1+|x|)^b dx)$, we have that
\begin{equation}\label{nothing2}
 |\hat v_\gamma(k)-\hat v_\gamma (0)|= \left| \frac{1}{\sqrt{2\pi}} \int_{\RE} (e^{ikx}-1) v_\gamma(x) dx  \right|  \leq |k|^{\bb} \frac{1}{\sqrt{2\pi}} \int_{\RE} |x|^{\bb}|v_\ga(x)| dx \leq C |k|^{\bb}. 
\end{equation}
 Moreover,  $\alpha_\gamma = \sqrt{2\pi} \hat v_\gamma(0)$ and the function $\tau_\gamma (\lambda)$ satisfies the identity 
\begin{equation}\label{harlem2}
\tau_\gamma(\lambda) = \frac{\hat v_\gamma(0)}{\sqrt{2\pi}} - \frac{1}{\sqrt{2\pi}}\int_{\RE} \frac{\hat v_\gamma(0)}{p^2/(2m_\gamma)+\lambda}\tau_\gamma(\lambda) dp .
\end{equation}
By taking the difference of Eqs. \eqref{t} and \eqref{harlem2} one has 
\begin{equation}\label{hawaii}
\begin{aligned}
\t_\gamma^\ve(\lambda;k,k') - \tau_\gamma(\lambda) =&  \frac{\hat v_\gamma(\ve(k-k'))-\hat v_\gamma(0)}{\sqrt{2\pi}} - \frac{1}{\sqrt{2\pi}}\int_{\RE} \frac{\hat v_\gamma(\ve(k-p))-\hat v_\gamma(0)}{p^2/(2m_\gamma)+\lambda}\tau_\gamma(\lambda) dp \\ 
& - \frac{1}{\sqrt{2\pi}}\int_{\RE} \frac{\hat v_\gamma(\ve(k-p))}{p^2/(2m_\gamma)+\lambda}(\t_\gamma^\ve(\lambda;p,k')-\tau_\gamma(\lambda)) dp.
\end{aligned}
\end{equation}
By using the fact that $\hat{v}_\gamma$ is bounded and  Bound \eqref{nothing2} in Eq. \eqref{hawaii}, we have that  
\begin{equation*}
\begin{aligned}
\frac{|\t_\gamma^\ve(\lambda;k,k') - \tau_\gamma(\lambda)|}{|k|^{\bb}+|k'|^{\bb}+1} \leq  & C\ve^{\bb} + \frac{C \ve^{\bb} |\tau_\gamma(\lambda)|}{|k|^{\bb}+|k'|^{\bb}+1} \int_{\RE} \frac{|k|^{\bb}+|p|^\bb}{p^2/(2m_\gamma)+\lambda} dp\\ 
&+\frac{C}{|k|^{\bb}+|k'|^{\bb}+1} \int_{\RE} \frac{|p|^{\bb}+|k'|^{\bb}+1}{p^2/(2m_\gamma)+\lambda}\frac{|\t_\gamma^\ve(\lambda;p,k')-\tau_\gamma(\lambda)|}{|p|^{\bb}+|k'|^{\bb}+1} dp.
\end{aligned}
\end{equation*}
We note that 
\[
\frac{1}{|k|^{\bb}+|k'|^{\bb}+1} \int_{\RE} \frac{|p|^{\bb}+|k|^{\bb}+|k'|^{\bb}+1}{p^2/(2m_\gamma)+\lambda}dp  \leq \frac{C}{\la^\frac12} + \frac{C}{\la^{\frac12-\frac{b}2}}. 
\]
Hence, for $\lambda > 1$, one has  
\[
\sup_{k,k'\in\RE}\frac{|\t_\gamma^\ve(\lambda;k,k') - \tau_\gamma(\lambda)|}{|k|^{\bb}+|k'|^{\bb}+1} \leq   C\ve^{\bb} \left(1+\frac{|\tau_\gamma (\lambda)|}{\lambda^{\frac12-\frac{b}2}}\right) +  
\frac{C}{\lambda^{\frac12-\frac{\bb}2}} \sup_{p,k'\in\RE}\frac{|\t_\gamma^\ve(\lambda;p,k')-\tau_\gamma(\lambda)|}{|p|^{\bb }+|k'|^{\bb}+1}.
\]
The latter bound implies Bound \eqref{tvet}, by  Rem. \ref{r:tau} and  up to taking $\lambda$ large enough.
\end{proof}
\begin{proof}[{\bf Proof of Theorem \ref{t:main}}] 
In the proof of the theorem we set $b=s$ if $0<s<1$, if $s\geq 1$ one can chooses any $b \in(0,1)$. 

We fix $\la_0$ such that 
\[\min\{\inf \sigma(\HH^\ve),\inf \sigma(\HH)\} > -\lambda_0 \qquad  \forall \ve>0,   \]
and we prove that, for $\lambda>\la_0 $ large enough, one has 
\begin{equation*}
\lim_{\ve\to 0} \|(\HH^\ve  + \lambda)^{-1}- (\HH +\lambda)^{-1}\|_{\Bnorm(L^2(\RE^2))} = 0,
\end{equation*}
where $ \|\cdot\|_{\Bnorm(L^2(\RE^2))} $ denotes the usual norm for bounded operators in $L^2(\RE^2)$. 
The convergence of $(\HH^\ve  - z)^{-1}$ to $(\HH  - z)^{-1}$ for any $z\in\CO\backslash\RE$  follows from the identity 
\[(\HH^\ve  - z)^{-1} - (\HH  - z)^{-1} = \frac{\HH^\ve  +\lambda}{\HH^\ve  - z} \left[ (\HH^\ve  +\la)^{-1} - (\HH  +\la)^{-1}\right]\frac{\HH +\la }{\HH  - z},\]
see for example \cite[Lem. 2.6.1]{davies96}. Since they are unitarily equivalent we can estimate the norm of $\hat \RR^\ve(\la) - \hat \RR(\la)$ where $ \hat \RR(\la)$ is the conjugate
of  $\RR(\la)$ through Fourier transform.
Taking into account Eqs.  \eqref{horse} and \eqref{eq:psi} we have that 
\[
(\hat \RR^\ve(\la) - \hat \RR(\la)) f = \hat \RR_0(\la)\sum_{m} \rho^{(m),\ve}(\la)   -  \hat\GG(\lambda)\hat\qq(\la),
\]
where $\hat\GG(\lambda) =(\hat G_{23}(\lambda),\hat G_{31}(\lambda),\hat G_{12}(\lambda))$, and $\hat\qq=(\hat q^{(1)},\hat q^{(2)},\hat q^{(3)})$. 
Taking into account the explicit form of the resolvent  $\hat \RR_0(\la)$ in the $p$-coordinates, see, e.g., Eq. \eqref{change4}, and by Eq.  \eqref{file} (together with the definition of $\xi^{(\ell)}(\la)$) we have
\[\begin{aligned}
& (\hat \RR^\ve(\la) - \hat \RR(\la)) f(p_2,p_1) = \\
= & 
\frac{\rho^{(1),\ve}(\la;p_2,p_1) - \xi^{(1)}(\la;p_1)}{\frac{p_2^2}{2m_{23}} + \frac{p_2 \cdot p_1}{m_3} + \frac{p_1^2}{2m_{13}}+\la} 
+
\frac{\rho^{(2),\ve}(\la;p_3,p_2) - \xi^{(2)}(\la;p_2)}{\frac{p_3^2}{2m_{31}} + \frac{p_3 \cdot p_2}{m_1} + \frac{p_2^2}{2m_{21}}+\la} 
+\frac{\rho^{(3),\ve}(\la;p_1,p_3) - \xi^{(3)}(\la;p_3)}{\frac{p_1^2}{2m_{12}} + \frac{p_1 \cdot p_3}{m_2} + \frac{p_3^2}{2m_{32}}+\la} 
\end{aligned}
\]
where $p_3 = -p_1-p_2$. Hence,
\begin{multline}
\label{gl}
\| (\hat \RR^\ve(\la) - \hat \RR(\la)) f \|^2_{L^2(\RE^2)} 
 \leqslant C 
\int_{\RE^2} dq\,dp  \left[\frac{\left|\rho^{(1),\ve}(\la;q,p) - \xi^{(1)}(\la;p)\right|^2}{\left|\frac{q^2}{2m_{23}} + \frac{q \cdot p}{m_3} + \frac{p^2}{2m_{13}}+\la\right|^2}  \right. \\
 \left.+
 \frac{\left|\rho^{(2),\ve}(\la;q,p) - \xi^{(2)}(\la;p)\right|^2}{\left|\frac{q^2}{2m_{31}} + \frac{q \cdot p}{m_1} + \frac{p^2}{2m_{21}}+\la\right|^2} 
+
 \frac{\left|\rho^{(3),\ve}(\la;q,p) - \xi^{(3)}(\la;p)\right|^2}{\left|\frac{q^2}{2m_{12}} + \frac{q \cdot p}{m_2} + \frac{p^2}{2m_{32}}+\la\right|^2}\right].
\end{multline}
We note the chain of  inequalities  
\begin{align}
\int_{\RE^2} dq\,dp\, \frac{\left|\rho^{(1),\ve}(\la;q,p) - \xi^{(1)}(\la;p)\right|^2}{\left|\frac{q^2}{2m_{23}} + \frac{q \cdot p}{m_3} + \frac{p^2}{2m_{13}}+\la\right|^2}  
\leq & C_{12}^2 \int_{\RE} dq\, \frac{1}{(q^2+C_{12}\la)^{2-\bb}} \int_{\RE} dp\,\frac{ \left|\rho^{(1),\ve}(\la;q,p) - \xi^{(1)}(\la;p)\right|^2}{(q^2+p^2+C_{12}\la)^{\bb}}  \nonumber \\ 
\leq & \frac{C}{\la^{\frac32-\bb}} \sup_{q\in\RE}\int_{\RE} dp\,\frac{\left|\rho^{(1),\ve}(\la;q,p) - \xi^{(1)}(\la;p)\right|^2 }{(q^2+p^2+C_{123}\la)^{\bb}}  , \label{atom}
\end{align}
for any $0<b<3/2$. Here we used Rem. \ref{r:away2} and the trivial inequality 
\[\frac{1}{(q^2+p^2+C_{12}\la)^{2}}\leq \frac{1}{(q^2+C_{12}\la)^{2-\bb}}\frac{1}{(q^2+p^2+C_{12}\la)^{\bb}},\] \\
and defined $C_{123} $ to be the least of $C_{12}$, $C_{23}$ and $C_{31}$.
Two analogous inequalities hold true for the terms involving   $\rho^{(2),\ve} - \xi^{(2)}$ and $\rho^{(3),\ve} - \xi^{(3)}$. 

Hence it is sufficient to prove
\[
\lim_{\ve\to 0}
\sup_{q\in\RE}\int_{\RE} dp\,\frac{\left|\rho^{(\ell),\ve}(\la;q,p) - \xi^{(\ell)}(\la;p)\right|^2 }{(q^2+p^2+C_{123}\la)^{\bb}} =0 \qquad \ell = 1,2,3. 
\]
Using Eqs. \eqref{rho1eq} and \eqref{eq:stm} we have
\begin{equation*}
\begin{aligned}
&\left|\rho^{(1),\ve}(\la;q,p) - \xi^{(1)}(\la;p)\right|^2   \\
 \leqslant &C \left[ \lf|
 \int_{\RE} dq' \, \frac{t_{23}^\ve(\lambda+\frac{p^2}{2\mu_1};-q-\frac{m_2}{m_2+m_3}p,-q'-\frac{m_2}{m_2+m_3}p)-\tau_{23}\bigl(\lambda+\frac{p^2}{2\mu_1}\bigr)}{\frac{{q'}^2}{2m_{23}}+\frac{q'\cdot p}{m_3}+\frac{p^2}{2m_{31}}+\lambda}f(q',p) \ri|^2 \right.\\
&+\lf|
\int_{\RE} dq'\,\frac{t_{23}^\ve(\lambda+\frac{p^2}{2\mu_1};-q-\frac{m_2}{m_2+m_3}p,-q'-\frac{m_2}{m_2+m_3}p)- \tau_{23}\bigl(\lambda+\frac{p^2}{2\mu_1}\bigr)}{\frac{{q'}^2}{2m_{23}}+\frac{q'\cdot p}{m_3}+\frac{p^2}{2m_{31}}+\lambda}\xi^{(2)}(\la;q')  \ri|^2  \\
&+\lf|
 \int_{\RE} dq'\,\frac{t_{23}^\ve(\lambda+\frac{p^2}{2\mu_1};-q-\frac{m_2}{m_2+m_3}p,q'+\frac{m_3}{m_2+m_3}p)-\tau_{23}\bigl(\lambda+\frac{p^2}{2\mu_1}\bigr)}{\frac{{q'}^2}{2m_{23}}+\frac{q'\cdot p}{m_2}+\frac{p^2}{2m_{12}}+\lambda}\xi^{(3)}(\la;q') \ri|^2  \\
&+ \lf|
 \int_{\RE} dq' \, \frac{t_{23}^\ve(\lambda+\frac{p^2}{2\mu_1};-q-\frac{m_2}{m_2+m_3}p,-q'-\frac{m_2}{m_2+m_3}p)}{\frac{{q'}^2}{2m_{23}} + \frac{q' \cdot p}{m_3} + \frac{p^2}{2m_{31}}+\la} (\rho^{(2),\ve}(\la;-p-q',q') - \xi^{(2)}(\la;q'))  \ri|^2\\
&\left. +  \lf|
    \int_{\RE} dq' \, \frac{t_{23}^\ve(\lambda+\frac{p^2}{2\mu_1};-q-\frac{m_2}{m_2+m_3}p,q'+\frac{m_3}{m_2+m_3}p)}{\frac{{q'}^2}{2m_{23}} + \frac{  q' \cdot p}{m_2} + \frac{p^2}{2m_{12}}+\la} (\rho^{(3),\ve}(\la;p,q') - \xi^{(3)}(\la;q'))  \ri|^2 \right]
\end{aligned}
\end{equation*}
The latter inequality, together with Lemma \ref{l:tvet} and  Rem. \ref{r:away2} (setting, as above, $C_{123} $ to be the least of $C_{12}$, $C_{23}$ and $C_{31}$) give
\begin{align}
&\int_{\RE} dp\,\frac{\left|\rho^{(1),\ve}(\la;q,p) - \xi^{(1)}(\la;p)\right|^2 }{(q^2+p^2+C_{123}\la)^{\bb}}   \label{norma1}\\
 \leqslant  & C \ve^{2b} \int_{\RE} \frac{dp}{{(q^2+p^2+C_{123}\la)^{\bb}}} \left(
 \int_{\RE} dq' \, \frac{|p|^{\bb} + |q|^{\bb} + |q'|^{\bb}    +1}{q'^2+p^2+C_{123}\la}\lf|f(q',p) \ri| \right)^2  \label{pezzo1}\\
&+C \ve^{2b} \int_{\RE} \frac{dp}{{(q^2+p^2+C_{123}\la)^{\bb}}} \left(
\int_{\RE} dq'\,\frac{|p|^{\bb} + |q|^{\bb} + |q'|^{\bb}    +1}{q'^2+p^2+C_{123}\la}  \lf|\xi^{(2)}(\la;q')  \ri|  \right)^2 \label{pezzo2}\\
&+C \ve^{2b} \int_{\RE} \frac{dp}{{(q^2+p^2+C_{123}\la)^{\bb}}} \left(
 \int_{\RE} dq'\,\frac{|p|^{\bb} + |q|^{\bb} + |q'|^{\bb}    +1}{q'^2+p^2+C_{123}\la} \left|\xi^{(3)}(\la;q') \ri| \right)^2  \label{pezzo3} \\
&+C \int_{\RE} \frac{dp}{{(q^2+p^2+C_{123}\la)^{\bb}}} \left(
 \int_{\RE} dq' \, \frac{\left| \rho^{(2),\ve}(\la;-p-q',q') - \xi^{(2)}(\la;q')  \ri|}{q'^2+p^2+C_{123}\la}
 \right)^2 \label{pezzo4}\\
&+ C \int_{\RE} \frac{dp}{{(q^2+p^2+C_{123}\la)^{\bb}}} \left(
    \int_{\RE} dq' \, \frac{\left| \rho^{(3),\ve}(\la;p,q') - \xi^{(3)}(\la;q')  \right|}{q'^2+p^2+C_{123}\la}  \right)^2 \label{pezzo5}
\end{align}

Using Cauchy-Schwarz inequality, and the trivial inequality \[(|p|^{\bb} + |q|^{\bb} + |q'|^{\bb}    +1)^2 \leq C (p^{2\bb} + q^{2\bb} + q'^{2\bb}    +1),\] the term in Eq. \eqref{pezzo1} can be estimated by
\begin{align*}
&   \ve^{2b} C\int_{\RE} \frac{dp}{{(q^2+p^2+C_{123}\la)^{\bb}}} \lf(  \int_{\RE} dq' \frac{p^{2\bb} + q^{2\bb} + q'^{2\bb}    +1         }{({q'}^2+p^2+C_{123}\la)^{2}}\ri)  \lf(  \int_{\RE} dq'  |f(q',p)|^2\ri) \\
&\leqslant  \ve^{2b} C\| f\|_{L^2 (\RE^2)}^2,
\end{align*}
where we used $ \frac{p^{2\bb} + q^{2\bb}}{(q^2+p^2+C_{123}\la)^{\bb}}\leq C $ and $ \frac{1}{(q^2+p^2+C_{123}\la)^{\bb}} \leq 1$ (for $\la$ large enough), which in turn imply 
\[
 \frac{1}{(q^2+p^2+C_{123}\la)^{\bb}} \int_{\RE} dq'  \frac{p^{2\bb} + q^{2\bb} + q'^{2\bb}    +1}{({q'}^2+p^2+C_{123}\la)^{2}}
 \leqslant C \int_{\RE} dq'  \frac{ q'^{2\bb}    +1         }{ ({q'}^2+C_{123}\la)^{2}} \leqslant C.
\]

Using Cauchy-Schwarz inequality, and Bound  \eqref{xi-bound}, the term in Eq. \eqref{pezzo2} can be estimated by
\begin{equation*}
\ve^{2b} C\int_{\RE} \frac{dp}{{(q^2+p^2+C_{123}\la)^{\bb}}} \lf(  \int_{\RE} dq' \frac{p^{2\bb} + q^{2\bb} + q'^{2\bb}+1}{({q'}^2+p^2+C_{123}\la)^{2}}\ri)  \lf(  \int_{\RE} dq'  |\xi^{(2)} (q')|^2\ri) \leqslant  \ve^{2b} C\| f\|_{L^2 (\RE^2)}^2
\end{equation*}
where we used 
\begin{multline*}
 \int_{\RE^2} dp\, dq'  \frac{p^{2\bb} + q^{2\bb} + q'^{2\bb}    +1         }{(q^2+p^2+C_{123}\la)^{\bb} ({q'}^2+p^2+C_{123}\la)^{2}} \\ 
  \leqslant 
  \int_{\RE^2} dp\, dq'  \frac{1}{({q'}^2+p^2+C_{123}\la)^{2}} +  \int_{\RE^2} dp\, dq'  \frac{q'^{2\bb}    +1}{(p^2+C_{123}\la)^{\bb+\frac12} ({q'}^2+C_{123}\la)^{\frac32}} \leqslant C.
\end{multline*}
The same estimate holds true  for the term in Eq.  \eqref{pezzo3}. 

Using Cauchy-Schwarz inequality, the term in Eq.  \eqref{pezzo4} can be estimated by
\begin{align*}
&C\int_{\RE} \frac{dp}{{(q^2+p^2+C_{123}\la)^{\bb}}} \lf(  \int_{\RE} dq' \frac{p^{2\bb}  + q'^{2\bb}    +(C_{123}\la)^{\bb}       }{({q'}^2+p^2+C_{123}\la)^{2}} \ri)
  \lf(  \int_{\RE} dq'\frac{ \lf|\rho^{(2),\ve}(\la;-p-q',q') - \xi^{(2)}(\la;q'))  \ri|^2}{ (q'^2+p^2+C_{123}\la)^{\bb} }  \ri) \\
& \leqslant \f{C}{\la}\sup_{p\in\RE} \int_{\RE} dq'\frac{ \lf|\rho^{(2),\ve}(\la;-p-q',q') - \xi^{(2)}(\la;q'))  \ri|^2}{ (q'^2+p^2+C_{123}\la)^{\bb} }  
\end{align*}
where we used
\begin{align*}
& \int_{\RE^2} dp \, dq'  \frac{p^{2\bb}  + q'^{2\bb}    +(C_{123}\la)^{\bb}          }{(q^2+p^2+C_{123}\la)^{\bb} ({q'}^2+p^2+C_{123}\la)^{2}} \\
& \leqslant \int_{\RE^2} dp \, dq'  \frac{1         }{ ({q'}^2+p^2+C_{123}\la)^{2}} +\int_{\RE^2} dp \, dq'  \frac{q'^{2\bb}             }{(p^2+C_{123}\la)^{\bb} ({q'}^2+p^2+C_{123}\la)^{2}} \\
&  \leqslant \f{C}{\la},
\end{align*}
which can be easily proved by scaling. In the same way, the term in Eq. \eqref{pezzo5} can be estimated by
\[
 \f{C}{\la}\sup_{p\in\RE} \int_{\RE} dq'\frac{ \lf|\rho^{(3),\ve}(\la;p,q') - \xi^{(3)}(\la;q')  \ri|^2}{ (q'^2+p^2+C_{123}\la)^{\bb} } .
\]
Therefore we obtain
\begin{multline} \label{penna}
\sup_{q\in\RE}\int_{\RE} dp\,\frac{\left|\rho^{(1),\ve}(\la;q,p) - \xi^{(1)}(\la;p)\right|^2 }{(q^2+p^2+C_{123}\la)^{\bb}}  \leqslant    \ve^{2b} C\| f\|_{L^2 (\RE^2)}^2 + \\
+ \f{C}{\la} \lf(  \sup_{q\in\RE} \int_{\RE} dp\frac{ \lf|\rho^{(2),\ve}(\la;-p-q,p) - \xi^{(2)}(\la;p))  \ri|^2}{ (q^2+p^2+C_{123}\la)^{\bb} }  +
\sup_{q\in\RE} \int_{\RE} dp\frac{ \lf|\rho^{(3),\ve}(\la;q,p) - \xi^{(3)}(\la;p)  \ri|^2}{ (q^2+p^2+C_{123}\la)^{\bb} } \ri).
\end{multline}
Two similar bounds with $|\rho^{(2),\ve}(\la;q,p) - \xi^{(2)}(\la;p)|$ and $|\rho^{(3),\ve}(\la;q,p) - \xi^{(3)}(\la;p)|$ at the l.h.s. are obtained by permutation of the indices. 

Estimate \eqref{penna} is not sufficient to close the proof since it  involves also terms containing the function  $\rho^{(\ell),\ve}(\la;-p-q,p)$. To obtain bounds on those terms we note that 
\begin{align}
&\left|\rho^{(1),\ve}(\la;-q-p,p) - \xi^{(1)}(\la;p)\right|^2   \leqslant \nonumber\\
 \leqslant & C\left[  \lf|
 \int_{\RE} dq' \, \frac{t_{23}^\ve(\lambda+\frac{p^2}{2\mu_1};q+\frac{m_3}{m_2+m_3}p,-q'-\frac{m_2}{m_2+m_3}p)-\tau_{23}\bigl(\lambda+\frac{p^2}{2\mu_1}\bigr)}{\frac{{q'}^2}{2m_{23}}+\frac{q'\cdot p}{m_3}+\frac{p^2}{2m_{31}}+\lambda}f(q',p) \ri|^2  \right.\nonumber \\
&+\lf|
\int_{\RE} dq'\,\frac{t_{23}^\ve(\lambda+\frac{p^2}{2\mu_1};q+\frac{m_3}{m_2+m_3}p,-q'-\frac{m_2}{m_2+m_3}p)- \tau_{23}\bigl(\lambda+\frac{p^2}{2\mu_1}\bigr)}{\frac{{q'}^2}{2m_{23}}+\frac{q'\cdot p}{m_3}+\frac{p^2}{2m_{31}}+\lambda}\xi^{(2)}(\la;q')  \ri|^2  \nonumber \\
&+ \lf|
 \int_{\RE} dq'\,\frac{t_{23}^\ve(\lambda+\frac{p^2}{2\mu_1};q+\frac{m_3}{m_2+m_3}p,q'+\frac{m_3}{m_2+m_3}p)-\tau_{23}\bigl(\lambda+\frac{p^2}{2\mu_1}\bigr)}{\frac{{q'}^2}{2m_{23}}+\frac{q'\cdot p}{m_2}+\frac{p^2}{2m_{12}}+\lambda}\xi^{(3)}(\la;q') \ri|^2   \nonumber\\
&+\lf|
 \int_{\RE} dq' \, \frac{t_{23}^\ve(\lambda+\frac{p^2}{2\mu_1};q+\frac{m_3}{m_2+m_3}p,-q'-\frac{m_2}{m_2+m_3}p)}{\frac{{q'}^2}{2m_{23}} + \frac{q' \cdot p}{m_3} + \frac{p^2}{2m_{31}}+\la} (\rho^{(2),\ve}(\la;-p-q',q') - \xi^{(2)}(\la;q'))  \ri|^2  \nonumber \\
&\left.+ \lf|
    \int_{\RE} dq' \, \frac{t_{23}^\ve(\lambda+\frac{p^2}{2\mu_1};q+\frac{m_3}{m_2+m_3}p,q'+\frac{m_3}{m_2+m_3}p)}{\frac{{q'}^2}{2m_{23}} + \frac{  q' \cdot p}{m_2} + \frac{p^2}{2m_{12}}+\la} (\rho^{(3),\ve}(\la;p,q') - \xi^{(3)}(\la;q'))  \ri|^2 \right] \nonumber
\end{align}
Repeating the same steps used  from  Eq. \eqref{norma1} to Eq. \eqref{penna}, one can see that  the estimate 
\begin{multline} \label{matita}
\sup_{q\in\RE}\int_{\RE} dp\,\frac{\left|\rho^{(1),\ve}(\la;-q-p,p) - \xi^{(1)}(\la;p)\right|^2 }{(q^2+p^2+C_{123}\la)^{\bb}}  \leqslant    \ve^{2b} C\| f\|_{L^2 (\RE^2)}^2 + \\
+ \f{C}{\la} \lf(  \sup_{q\in\RE} \int_{\RE} dp\frac{ \lf|\rho^{(2),\ve}(\la;-p-q,p) - \xi^{(2)}(\la;p))  \ri|^2}{ (q^2+p^2+C_{123}\la)^{\bb} }  +
\sup_{q\in\RE} \int_{\RE} dp\frac{ \lf|\rho^{(3),\ve}(\la;q,p) - \xi^{(3)}(\la;p)  \ri|^2}{ (q^2+p^2+C_{123}\la)^{\bb} } \ri)
\end{multline}
holds trues, and similar ones are obtained by permutation of the indices. 

Summing up over permutations of indices the estimates \eqref{penna} and \eqref{matita} we obtain
\begin{equation*}
\begin{aligned}
&\sum_{j=1}^3 \left(\sup_{q\in\RE}\int_{\RE} dp\,\frac{\left|\rho^{(j),\ve}(\la;-q-p,p) - \xi^{(j)}(\la;p)\right|^2 }{(q^2+p^2+C_{123}\la)^{\bb}} +
 \sup_{q\in\RE}\int_{\RE} dp\,\frac{\left|\rho^{(j),\ve}(\la;q,p) - \xi^{(j)}(\la;p)\right|^2 }{(q^2+p^2+C_{123}\la)^{\bb}}\right)  \\ 
 \leqslant &
 \ve^{2b} C\| f\|_{L^2 (\RE^2)}^2  \\
&+\f{C}{\la} 
\sum_{j=1}^3 \left( \sup_{q\in\RE}\int_{\RE} dp\,\frac{\left|\rho^{(j),\ve}(\la;-q-p,p) - \xi^{(j)}(\la;p)\right|^2 }{(q^2+p^2+C_{123}\la)^{\bb}} +
\sup_{q\in\RE}\int_{\RE} dp\,\frac{\left|\rho^{(j),\ve}(\la;q,p) - \xi^{(j)}(\la;p)\right|^2 }{(q^2+p^2+C_{123}\la)^{\bb}} 
\right) .
\end{aligned}
\end{equation*}
For $\la$ sufficiently large, the latter inequality implies 
\begin{multline*}
\sum_{j=1}^3\left( \sup_{q\in\RE}\int_{\RE} dp\,\frac{\left|\rho^{(j),\ve}(\la;-q-p,p) - \xi^{(j)}(\la;p)\right|^2 }{(q^2+p^2+C_{123}\la)^{\bb}} +\sup_{q\in\RE}\int_{\RE} dp\,\frac{\left|\rho^{(j),\ve}(\la;q,p) - \xi^{(j)}(\la;p)\right|^2 }{(q^2+p^2+C_{123}\la)^{\bb}}\right) \\ 
 \leqslant
\ve^{2b} C\| f\|_{L^2 (\RE^2)}^2.
\end{multline*}
Hence,  from Bounds \eqref{gl} and \eqref{atom} it follows that 
\begin{equation*}\label{gl2}
 \| (\hat \RR^\ve(\la) - \hat \RR(\la)) f \|^2_{L^2(\RE^2)} \leq \ve^{2b} C\| f\|_{L^2 (\RE^2)}^2
\end{equation*}
and the proof is concluded.
\end{proof}

\appendix

\section{Faddeev's equations \label{app:A}}
 For the convenience of the reader, in this section we shortly recall the derivation of Faddeev's equations, for more details we refer to Faddeev's book \cite{faddeev_book}. 
\subsection{Resolvent formulae and Faddeev equations}
We look for an equation for the resolvent of Hamiltonian \eqref{Hred}. We start with the resolvent identity and write 
\begin{equation}\label{rep1}
\RR^\ve(\lambda) = (\HH^\ve+\lambda)^{-1} =( \HH_0 +\sum_{\sigma} \VV_{\sigma}^\ve+\lambda)^{-1} = \RR_0(\lambda) + \sum_{\ell}\RR^{(\ell),\ve}(\lambda),
\end{equation}
where $\RR_0(\lambda) = ( \HH_0 +\lambda)^{-1}$ is the resolvent of the free Hamiltonian $\HH_0$, and 
\begin{equation}\label{perfect1}
\RR^{(\ell),\ve}(\lambda): = - \RR_0(\lambda) \VV_{\gamma}^\ve \RR^\ve(\lambda).
\end{equation}
On the other hand, again by the resolvent identity, one has  
\begin{equation}\label{perfect2}\RR^\ve(\lambda) =  \RR_\gamma^\ve(\lambda) -\RR_{\gamma}^\ve(\lambda) \sum_{\sigma\neq\gamma }\VV_\sigma^\ve \RR^\ve(\lambda),\end{equation}
where $\RR_\gamma^\ve(\lambda): =( \HH_0 + \VV_{\gamma}^\ve+\lambda)^{-1}$.  Note that  the Hamiltonian $\HH_\gamma^\ve:=\HH_0 + \VV_{\gamma}^\ve$ in the coordinates $(x_{\gamma},y_\ell)$ is factorized, because $\VV_\gamma^\ve  = \VV_\gamma^\ve (x_\gamma)$.
Plugging Eq. \eqref{perfect2} in  Eq. \eqref{perfect1} one ends up with 
\begin{equation}\label{nine1}
\RR^{(\ell),\ve}(\lambda) = - \RR_0(\lambda) \VV_{\gamma}^\ve\RR_\gamma^\ve(\lambda)  + \RR_0(\lambda) \VV_{\gamma}^\ve \RR_{\gamma}^\ve(\lambda) \sum_{\sigma\neq\gamma }\VV_\sigma^\ve \RR^\ve(\lambda).
\end{equation}
Next we define the operator 
\[\TT_\gamma^\ve(\lambda)  := \VV_\gamma^\ve -\VV_\gamma^\ve \RR_\gamma^\ve(\lambda)\VV_\gamma^\ve\] 
and note the identity 
\begin{equation}\label{nine2}
\RR_0(\lambda) \VV_{\gamma}^\ve\RR_\gamma^\ve(\lambda)  = \RR_0(\lambda) \TT_{\gamma}^\ve(\la)\RR_0(\lambda),
\end{equation}
which is a direct consequence of the resolvent identity $\RR_\gamma^\ve(\lambda) =  \RR_0 (\lambda) -\RR_{\gamma}^\ve (\lambda)\VV_\gamma^\ve \RR_0(\lambda)$. By using Eq. \eqref{nine2} in Eq. \eqref{nine1} we get  
\begin{equation}\label{rep2}
\RR^{(\ell),\ve}(\lambda) = - \RR_0(\lambda)  \TT_{\gamma}^\ve(\la)\RR_0(\lambda) - \RR_0(\lambda)  \TT_{\gamma}^\ve(\la)\sum_{m\neq \ell}\RR^{(m),\ve}(\lambda).
\end{equation}
%
%
By Eqs. \eqref{rep1} and \eqref{rep2},  we conclude that for any function $f\in L^2(\RE^2)$ one has 
\begin{equation}\label{resolvent}
\RR^\ve(\lambda)  f = \RR_0(\lambda) f + \sum_{m}g^{(m),\ve}(\lambda) \qquad \textrm{with} \quad g^{(\ell),\ve}(\la) =\RR^{(\ell),\ve}(\lambda)  f ,
\end{equation}
where the functions $g^{(\ell),\ve}(\la)$
must  solve the system of  equations 
\begin{equation}\label{faddeev}
g^{(\ell),\ve}(\la) =   - \RR_0(\lambda)  \TT_{\gamma}^\ve(\la)\RR_0(\lambda) f - \RR_0(\lambda)  \TT_{\gamma}^\ve(\la)\sum_{m\neq \ell}g^{(m),\ve}(\lambda).
\end{equation}

The system \eqref{faddeev} expresses a form of  Faddeev's equations \cite{faddeev-jetp61}. In our analysis we shall write  Faddeev's equations (in Fourier transform)   for the functions 
\begin{equation}\label{hand}
\rho^{(\ell),\ve}(\la) :=(\hat \HH_0+\la)\hat g^{(\ell),\ve}(\la).
\end{equation} By Eqs.  \eqref{resolvent} and \eqref{faddeev}, it is easy to convince oneself that the resolvent $\hat\RR^\ve(\la)$ can be written as in Eq. \eqref{horse} and that  the functions $\rho^{(\ell),\ve}(\la)$ must satisfy the system of equations obtained by Eq.  \eqref{shepard} through permutation of the indices.
  

\subsection{\label{ss:onepart}Reduced operators in terms of one-particle operators} In this section we derive a formula for the operators $\TT_\gamma^\ve(\la)$ in terms of one-particle operators. This formula allows to write the action of the operator $\TT_\gamma^\ve(\la)$ as in Eq. \eqref{electric} and to obtain Eq. \eqref{rho1eq}. 

We denote by lower case letters one-particle operators, i.e., operators acting on the space $L^2(\RE)$. In particular we shall use the notation 
\[\hh_{0}^{(\ell)} := -\frac1{2\mu_\ell}\Delta_{y_\ell} ,\quad \hh_{0}^{(\ell)} :L^2(\RE,dy_\ell) \to L^2(\RE,dy_\ell);\]
\[\hh_{0,\gamma} := -\frac1{2m_\gamma}\Delta_{x_\gamma} ,\quad \hh_{0,\gamma} :L^2(\RE,dx_\gamma) \to L^2(\RE,dx_\gamma);\]
$\rr_{0,\gamma}(\lambda) := (\hh_{0,\gamma} +\lambda)^{-1};$
\[\hh_{\gamma}^\ve := \hh_{0,\gamma} +\vv_\gamma^\ve,\quad \hh_{\gamma} :L^2(\RE,dx_\gamma) \to L^2(\RE,dx_\gamma);\]
$\rr_{\gamma}^\ve(\lambda) := (\hh_{\gamma}^\ve +\lambda)^{-1};$
here $\vv_\gamma^\ve$ is the two particle potential understood as a multiplication operator in $L^2(\RE,dx_\gamma)$. 

In particular we shall be interested in the one particle operator defined by the identity 
\begin{equation*}\label{tt1}
\tt_\gamma^\ve(\lambda) : = \vv_\gamma^\ve -\vv_\gamma^\ve \rr_\gamma^\ve(\lambda)\vv_\gamma^\ve,\quad \tt_\gamma^\ve(\lambda) :L^2(\RE,dx_\gamma) \to L^2(\RE,dx_\gamma).
\end{equation*}

We note that, by the resolvent identity $\rr_\gamma^\ve(\lambda) = \rr_{0,\gamma}(\lambda) - \rr_{0,\gamma}(\lambda)  \vv_\gamma^\ve \rr_\gamma^\ve(\lambda) $ 
one infers  that the operator $\tt(\lambda)$ satisfies the equation 
\begin{equation}\label{tt2}
\tt_\gamma^\ve(\lambda) = \vv_\gamma^\ve -\vv_\gamma^\ve \rr_{0,\gamma}(\lambda)\tt_\gamma^\ve(\lambda).
\end{equation}

Recalling that the Hamiltonian $\HH_\gamma^\ve$ is factorized in the coordinates $(x_\gamma,y_{\ell})$, one has that $\RR_\gamma^\ve (\lambda)$ can be formally written as 
\begin{equation*}\label{RR2}\RR_\gamma^\ve (\lambda) =  \rr_\gamma^\ve(\lambda + \hh_0^{(\ell)}),  \quad \RR_\gamma^\ve (\lambda):L^2(\RE^2,dx_\gamma dy_{\ell}) \to L^2(\RE^2,dx_\gamma dy_\ell) . 
\end{equation*}
Similarly 
\begin{equation}\label{TT2}
\TT_\gamma^\ve(\lambda) : = \VV_\gamma^\ve -\VV_\gamma^\ve \rr_\gamma^\ve(\lambda+\hh_0^{(\ell)})\VV_\gamma^\ve = \tt_\gamma^\ve(\lambda+\hh_0^{(\ell)}) :L^2(\RE^2,dx_\gamma dy_\ell) \to L^2(\RE^2,dx_\gamma dy_\ell) .
\end{equation} 
Identity \eqref{TT2} can be understood in Fourier transform, see Eqs. \eqref{tt} -  \eqref{electric}. 

\section{Some useful explicit formulae \label{app:B}}
In this section we collect several useful formulae, in particular for the operators appearing in Section \ref{s:Halde}. For sake of concreteness we write the formulae in the coordinates $(x_{23},y_1)$, and their conjugates $(k_{23},p_1)$, or in the coordinates $(p_2,p_1)$. Additional formulae are obtained by permutation of the indices or by change of variables. 

We remark that  the Fourier transform is defined so as to be unitary in $L^2(\RE^d)$. Explicitly, the Fourier transform in $L^2(\RE^d)$ is denoted by  ${\hat{}\mbox{} }$ and defined as  
\[\hat f(k) := \frac{1}{(2\pi)^{d/2}}  \int_{\RE^d} e^{-ikx} f(x) dx. \]
The inverse Fourier transform is given by 
\[\check f (x) :=  \frac{1}{(2\pi)^{d/2}}\int_{\RE^d} e^{ikx} f(k) dk. \]
Moreover
\[\widehat{(f*g)}(k) = (2\pi)^{d/2} \hat f (k) \, \hat g (k); \]
\[\widehat{(fg)}(k) = \frac{1}{(2\pi)^{d/2}} (\hat f *\hat g) (k);\]
and
\[(f,g)_{L^2(\RE^d)} = (\hat f,\hat g)_{L^2(\RE^d)}.  \]

We start by noticing that the Fourier transform of the operator $\breve G_{23}$, see Eq. \eqref{breveG}, is given by 
\begin{equation}\label{breveGfourier}
\widehat{\breve G_{23}}(\la)\hat f(p_1) = \frac{1}{\sqrt{2\pi}} \int_{\RE} dk_{23} \frac{1}{\frac{k^2_{23}}{2 m_{23}} + \frac{p_1^2}{2\mu_1} + \la} \hat f(k_{23},p_1).  \end{equation}
Hence, 
\begin{equation*}
\breve G_{23}(\la)f(y_1) = \frac{1}{\sqrt{2\pi}} \int_{\RE} dp_1 e^{iy_1p_1} \widehat{\breve G_{23}}(\la)\hat{f}(p_1). 
\end{equation*}

By taking the adjoint of $\breve G_{23}(\la)$, it is easy to convince oneself  that in Fourier transform the operator $G_{23}(\la)$ acts as the multiplication operator 
\begin{equation}\label{Gfourier}
\hat{G}_{23}(\la)\hat q (k_{23},p_1) =   \frac{1}{\sqrt{2\pi}}  \frac{1}{\frac{k^2_{23}}{2 m_{23}} + \frac{p_1^2}{2\mu_1} + \la} \hat q(p_1).
\end{equation}
Hence, 
\begin{align}
G_{23}(\la) q (x_{23},y_1)  = & \frac{1}{2\pi} \int_{\RE^2} dk_{23} \, dp_1 \, e^{ix_{23}k_{23}+iy_1p_1} \hat{G}_{23}(\la)\hat q (k_{23}, p_1)  
\label{pigs0} \\ 
= &  \frac{1}{2\sqrt{2\pi}} \int_{\RE} dp_1 e^{iy_1p_1}\, \sqrt{\frac{2m_{23}}{\frac{p_1^2}{2\mu_1}+\la}} \, e^{-|x_{23}| \sqrt{2m_{23} \left(\frac{p_1^2}{2\mu_1}+\la\right)}} \, \hat q(p_1),
\label{pigs}
\end{align}
where  the latter identity was obtained by integrating over $k_{23}$. 

Noticing that $M_{23,23}(\la)q(y_1) = G_{23}(\la) q (0,y_1)$ and taking into account Eq. \eqref{pigs} one infers that in Fourier transform $M_{23,23}(\la)$ acts as the multiplication operator 
\begin{equation}\label{night1}
\hat M_{23,23}(\la)\hat q(p_1) = \frac12 \sqrt{\frac{2m_{23}}{\frac{p_1^2}{2\mu_1}+\la}} \,  \hat q(p_1), 
\end{equation}
and $ M_{23,23}(\la) q(y_1) = \frac1{\sqrt{2\pi}} \int_{\RE} dp_1 \, e^{iy_1p_1}\, \hat M_{23,23}(\la)\hat q(p_1)$. 

To obtain the expression of $M_{23,12}(\la)$ in Fourier transform recall that, by changing the indices  in Eqs. \eqref{Gfourier} and \eqref{pigs0}, one has 
\begin{equation*}
G_{12}(\la) q (x_{12},y_3)  =  \frac{1}{2\pi} \int_{\RE^2} dk_{12} \, dp_3 \, e^{ix_{12}k_{12}+iy_3p_3} 
 \frac{1}{\sqrt{2\pi}}  \frac{1}{\frac{k^2_{12}}{2 m_{12}} + \frac{p_3^2}{2\mu_3} + \la} \hat q(p_3).
\end{equation*} 
In the Jacobi coordinates $(x_{23},y_{1})$ (and the corresponding conjugate set $(k_{23},p_1)$), one has that $(x_{12}k_{12}+y_3p_3)|_{\pi_{23}} = y_1 p_1$. Since $M_{23,12}(\la)q(y_1) = G_{12}(\la) q|_{\pi_{23}}(y_1)$ and by the change of  variables $(k_{12},p_3)\to (k_{23},p_1)$ in the integral above, one obtains  
\begin{equation*}
M_{23,12}(\la)q(y_1) = \frac{1}{2\pi} \int_{\RE^2} dk_{23} \, dp_1 \, e^{iy_1p_1} 
 \frac{1}{\sqrt{2\pi}}  \frac{1}{\frac{k^2_{23}}{2 m_{23}} + \frac{p_1^2}{2\mu_1} + \la} \hat q(p_3(k_{23},p_1)),
\end{equation*}
note that  $p_3$ in the function  $\hat q$ must be understood as a function of the variables $(k_{23},p_1)$, as in Eq. \eqref{change2}. By the change of variables $(k_{23},p_1)\to (p_3,p_1)$ it follows that  
\begin{equation*}
M_{23,12}(\la)q(y_1) = \frac{1}{2\pi} \int_{\RE^2} dp_3 \, dp_1 \, e^{iy_1p_1} 
 \frac{1}{\sqrt{2\pi}}  \frac{1}{\frac{p_3^2}{2m_{23}}+\frac{p_3\cdot p_1}{m_2}+\frac{p_1^2}{2m_{12}}+\lambda}\hat q(p_3), 
\end{equation*}
hence in Fourier transform  $M_{23,12}(\la)$ acts as 
\begin{equation}\label{night2}
\hat M_{23,12}(\la)\hat q(p_1)  =  \frac{1}{2\pi} \int_{\RE} dp_3 \frac{1}{\frac{p_3^2}{2m_{23}}+\frac{p_3\cdot p_1}{m_2}+\frac{p_1^2}{2m_{12}}+\lambda}\hat q(p_3). 
\end{equation}
In a similar way one obtains 
\begin{equation}\label{night3}
\hat M_{23,31}(\la)\hat q(p_1)  =  \frac{1}{2\pi} \int_{\RE} dp_2 \frac{1}{\frac{p_2^2}{2m_{23}}+\frac{p_2\cdot p_1}{m_3}+\frac{p_1^2}{2m_{31}}+\lambda}\hat q(p_2). 
\end{equation}

We conclude this section by noting that in the coordinates $(p_2,p_1)$ the operators  $\widehat{\breve G_{23}}(\la)$ and  $\hat{G}_{23}(\la)$, see Eqs. \eqref{breveGfourier} and \eqref{Gfourier}, are  given by 
\begin{equation*}
\widehat{\breve G_{23}}(\la)\hat f(p_1) = \frac{1}{\sqrt{2\pi}} \int_{\RE} dp_2  \frac{1}{\frac{p_2^2}{2m_{23}} + \frac{p_2 \cdot p_1}{m_3} + \frac{p_1^2}{2m_{13}}+\la}  \hat f(p_2 ,p_1). 
\end{equation*}
and 
\begin{equation}\label{file}
\hat{G}_{23}(\la)\hat q (p_2,p_1) =   \frac{1}{\sqrt{2\pi}}  \frac{1}{\frac{p_2^2}{2m_{23}} + \frac{p_2 \cdot p_1}{m_3} + \frac{p_1^2}{2m_{13}}+\la} \hat q(p_1), 
\end{equation}
here with a slight abuse of notation we used the same symbols to denote the function $\hat f$ with the same symbol both in coordinates $(k_{23},p_1)$ and $(p_2,p_1)$. 

Similar identities are obtained by changes of variables and permutations of the indices.

\end{document}